\documentclass[11pt]{article}
\usepackage{amssymb}
\usepackage{amsfonts}
\usepackage{amsmath}
\usepackage{multirow}
\usepackage{latexsym}
\usepackage{epsfig}
\usepackage{bm}
\usepackage{times}
\usepackage{color}

\usepackage{url,amsmath,amsthm,amssymb,epsfig}
\usepackage[margin=.8in]{geometry}

\newcommand{\ignore}[1]{}

\makeatletter
\renewcommand{\subsection}{\@startsection{subsection}{2}{0pt}{-12pt}{-5pt}{\normalsize\textbf}}
\renewcommand{\subsubsection}{\@startsection{subsubsection}{3}{0pt}{-12pt}{-5pt}{\normalsize\textbf}}
\makeatother

\newtheorem{claim}{Claim}

\newtheorem{lemma}[claim]{Lemma}
\newtheorem{observation}[claim]{Observation}
\newtheorem{theorem}{Theorem}
\newtheorem{definition}{Definition}
\newtheorem{corollary}[claim]{Corollary}

\newtheorem{fact}[claim]{Fact}

\newcommand{\dtv}{d_{\mathrm TV}}
\newcommand{\dk}{d_{\mathrm K}}

\newcommand{\eps}{\epsilon}
\newcommand{\littlesum}{\mathop{{\textstyle \sum}}}

\newcommand{\poly}{\mathrm{poly}}

\newcommand{\eqdef}{\stackrel{\textrm{def}}{=}}
\newcommand{\polylog}{\mathrm{polylog}}

\newcommand{\wh}[1]{{\widehat{#1}}}

    \newcommand{\rnote}[1]{}
    \newcommand{\inote}[1]{}
   
   \newcommand{\cnote}[1]{\footnote{\textsc{ [[Costis: {#1}\textsc ]] }}}
   \newcommand{\todo}[1]{\textsc{ [[[XXX TO DO:{~~#1}\textsc XXX]]] }}

\newcommand{\RR}{\mathbb{R}}
\newtheorem{prop}[claim]{Proposition}

\begin{document}

\title{Testing $k$-Modal Distributions:  Optimal Algorithms via Reductions}

\author{Constantinos Daskalakis\thanks{{\tt costis@csail.mit.edu.}  Research supported by NSF CAREER award CCF-0953960 and by a
Sloan Foundation Fellowship.}\\
MIT\\
\and
Ilias Diakonikolas\thanks{
{\tt ilias@cs.berkeley.edu}.  Research supported by a Simons Foundation Postdoctoral Fellowship. Some of this work
was done while at Columbia University, supported by NSF grant CCF-0728736, and by an Alexander S. Onassis Foundation
Fellowship.}\\
UC Berkeley\\
\and
Rocco A. Servedio
\thanks{{\tt rocco@cs.columbia.edu}.  Supported by NSF grants CCF-0347282 and CCF-0523664.}\\
Columbia University\\
\and
Gregory Valiant\thanks{{\tt gregory.valiant@gmail.com}.  Supported by an NSF graduate research fellowship.}\\
UC Berkeley\\
\and
Paul Valiant\thanks{{\tt pvaliant@gmail.com}.  Supported by an NSF postdoctoral research fellowship.}\\
UC Berkeley\\
}

\setcounter{page}{0}

\maketitle

\thispagestyle{empty}

\begin{abstract}

We give highly efficient algorithms, and almost matching lower bounds, for a range of basic statistical problems that involve testing and estimating the $L_1$ (total variation) distance between two $k$-modal distributions $p$ and $q$ over the discrete domain $\{1,\dots,n\}$.
More precisely, we consider the following four problems:  given sample access to an unknown $k$-modal
distribution $p$,

\medskip

\noindent \textsc{ Testing identity to a known or unknown distribution:}
\begin{enumerate}

\vspace{-.2cm} \item Determine whether $p = q$ (for an explicitly given $k$-modal distribution $q$) versus
$p$ is $\eps$-far from $q$;

\vspace{-.2cm} \item Determine whether $p=q$ (where $q$ is available via sample access) versus
$p$ is $\eps$-far from $q$;
\end{enumerate}

\noindent \textsc{ Estimating $L_1$ distance (``tolerant testing'') against a known or unknown distribution:}

\begin{enumerate}
\vspace{-.2cm} \item [3.] Approximate $d_{TV}(p,q)$ to within additive $\epsilon$ where $q$ is an explicitly
given $k$-modal distribution $q$;

\vspace{-.2cm} \item [4.] Approximate $d_{TV}(p,q)$ to within additive $\epsilon$ where $q$ is available via sample access.
\end{enumerate}

\noindent

For each of these four problems we give sub-logarithmic sample algorithms, that we show are tight up to additive $\poly(k)$ and multiplicative $\polylog\log n+\polylog k$ factors.
Thus our bounds significantly improve the previous results of \cite{BKR:04}, which were for testing identity of distributions (items (1) and (2) above) in the special cases
$k=0$ (monotone distributions) and $k=1$ (unimodal distributions) and required $O((\log n)^3)$ samples.

\ignore{

only via samples, our algorithms require
$k^2 \cdot \poly(1/\eps) + \tilde{O}((k\log n)^{2/3}) \cdot \poly(1/\eps)$ samples.\footnote{We
write $\tilde{O}(\cdot)$ to hide factors which are polylogarithmic
in the argument to $\tilde{O}(\cdot)$;
thus for example $\tilde{O}(a \log b)$ denotes a quantity
which is $O((a \log b) \cdot  (\log (a \log b))^c)$ for some
absolute constant $c$.}
When $q$ is known and $p$ is available only via samples, our
algorithms require $k^2 \cdot \poly(1/\eps) + \tilde{O}((k\log n)^{1/2}) \cdot \poly(1/\eps)$
samples from $p$.  These bounds are nearly optimal, as we give an $\Omega((k \log n)^{1/2})$ lower bound for the case where $q$ is known and an $\Omega((k \log n)^{2/3})$ lower bound for the
case where $q$ is unknown.

}

As our main conceptual contribution, we introduce a new reduction-based approach for distribution-testing
problems that lets us obtain all the above results in a unified way.  Roughly
speaking, this approach enables us to transform various distribution testing problems for $k$-modal
distributions over $\{1,\dots,n\}$ to the corresponding distribution testing problems for unrestricted
distributions over a much smaller domain $\{1,\dots,\ell\}$ where $\ell = O(k \log n).$

\end{abstract}

\newpage

\section{Introduction} \label{sec:intro}

Given samples from a pair of unknown distributions, the problem of ``identity
testing''---that is, distinguishing whether the two distributions are \emph{the
same} versus significantly different---and, more generally, the problem of
estimating the $L_1$ distance between the distributions, is perhaps \emph{the} most
fundamental statistical task.  Despite a long history of study, by both the
statistics and computer science communities, the sample complexities of these basic
tasks were only recently established.  Identity testing, given samples from a pair
of distributions of support $[n]$, can be done using $\tilde{O}(n^{2/3})$ samples
\cite{BFR+:00}, and this upper bound is optimal up to $\polylog(n)$ factors~\cite{PValiant:08}. 
Estimating the $L_1$ distance (``tolerant testing'')
between distributions of support $[n]$ requires $\Theta(n/ \log n)$
samples, and this is tight up to constant factors~\cite{ValiantValiant:11,ValiantValiant:11focs}. 
The variants of these problems when one of the two distributions is explicitly given
require $\widetilde{\Theta}(\sqrt{n})$ samples for identity testing
\cite{BFFKRW:01} and $\Theta(n/\log n)$ samples for $L_1$ distance estimation
\cite{ValiantValiant:11,ValiantValiant:11focs} respectively.

While it is surprising that these tasks can be performed using a sublinear number of samples,
for many real-world applications using $\sqrt{n}$, $n^{2/3}$, or $\frac{n}{\log n}$ samples is still impractical. As these bounds characterize worst-case instances, one might hope that drastically better performance may be possible for many settings typically encountered in practice.
Thus, a natural research direction, which we pursue in this paper, is to understand
how structural properties of the distributions in question may be leveraged to yield improved
sample complexities.

In this work we consider monotone, unimodal, and more generally $k$-modal
distributions.  Monotone, unimodal, and bimodal distributions abound in the natural
world.  The distribution of many measurements---heights or weights of members of a population, concentrations of various chemicals in cells, parameters of many atmospheric phenomena--often belong to this class of distributions.  Because of their ubiquity, much work in the natural sciences rests on the analysis of such distributions (for example, on November 1, 2011 a Google Scholar search for the exact phrase ``bimodal distribution'' in the bodies of papers returned more than 90,000 hits).  Though perhaps not as pervasive, $k$-modal distributions
for larger values of $k$ commonly arise as mixtures of unimodal distributions and are natural
objects of study. On the theoretical side, motivated by the many applications, monotone, unimodal, and $k$-modal distributions have been intensively studied in the
probability and statistics literatures for decades, see e.g.
\cite{Grenander:56,PrakasaRao:69,BBBB:72,CKC:83,Groeneboom:85,Birge:87,Birge:87b,Kemperman:91,Fougeres:97,ChanTong:04,JW:09}.

\subsection{Our results.}

Our main results are algorithms, and nearly matching lower bounds, that give a complete
picture of the sample complexities of identity testing and estimating $L_1$
distance for monotone and $k$-modal distributions.  We obtain such results  in both the setting where the two
distributions are given via samples, and the setting where one of the distributions
is given via samples and the other is described explicitly.

All our results have the nature of a reduction:   performing these tasks on $k$-modal distributions
over $[n]$ turns out to have almost exactly the same sample complexities as performing
the corresponding tasks on \emph{arbitrary} distributions over $[k \log
n]$. For any small constant $k$ (or even $k=O((\log n)^{1/3})$) and arbitrarily small constant $\epsilon$, all our results are tight to within either $\polylog \log n$ or $\polylog\log\log n$ factors. See Table~1 for the new sample complexity upper and lower
bounds for the monotone and $k$-modal tasks; see Section~\ref{sec:prelims} for the (exponentially higher) sample complexities of the general-distribution tasks on which our results rely.
While our main focus is on sample complexity rather than running time, we note that all of our algorithms
run in $\poly(\log n, k, 1/\eps)$ bit operations (note that even reading a single sample from a distribution
over $[n]$ takes $\log n$ bit operations).

We view the equivalence between the sample complexity of each of the above tasks on a monotone or unimodal distribution of domain $[n]$ and the sample complexity of the same task on an unrestricted distribution of domain $[\log n]$ as a surprising result, because such an equivalence \emph{fails} to hold for related estimation tasks. For example, consider the task of distinguishing whether a distribution on $[n]$ is uniform versus far from uniform.  For general distributions this takes $\Theta(\sqrt{n})$ samples, so one might expect the corresponding problem for monotone distributions to need $\sqrt{\log n}$ samples; in fact, however, one can test this with a \emph{constant} number of samples, by simply comparing the empirically observed probability masses of the left and right halves
of the domain.  An example in the other direction is the problem of finding a constant additive estimate for the entropy of a distribution. On domains of size $[n]$ this can be done in $\frac{n}{\log n}$ samples, and thus one might expect to be able to estimate entropy for monotone distributions on $[n]$ using $\frac{\log n}{\log\log n}$ samples. Nevertheless, it is not hard to see that  $\Omega(\log^2 n)$ samples are required.

\begin{table}[t]
\begin{center}
\begin{tabular}{|l|c|c|c|}%
\hline \bf Testing problem & \bf Our upper bound & \bf Our lower bound
\\\hline\hline

\multicolumn{3}{|c|}{{\bf $\bm{p,q}$ are both monotone:}}\\\hline

Testing identity, $q$ is known: &
$O\left(\left(\log n\right)^{1/2}\left(\log\log n\right) \cdot \eps^{-5/2}\right)$ & $\Omega\left(\left(\log n\right)^{1/2}\right)$ \\\hline

Testing identity, $q$ is unknown: & $O\left(\left(\log n\right)^{2/3} \cdot \left(\log\log n\right) \cdot
\eps^{-10/3}\right)$ & $\Omega\left(\left({\frac {\log n}{\log \log n}}\right)^{2/3}\right)$
\\\hline

Estimating $L_1$ distance, $q$ is known:  &
$O\left({\frac {\log n} {\log \log n}} \cdot \eps^{-3}\right)$ &
$\Omega\left({\frac {\log n} {\log \log n \cdot \log \log \log n}}\right)$ \\ \hline

Estimating $L_1$ distance, $q$ is unknown:& $O\left({\frac {\log n} {\log \log n}} \cdot \eps^{-3} \right)$ &
$\Omega\left({\frac {\log n} {\log \log n \cdot \log \log \log n}}\right)$  \\\hline

\multicolumn{3}{|c|}{{\bf $\bm{p,q}$ are both $\bm k$-modal:}}\\\hline

Testing identity, $q$ is known: &
$O\left({\frac {k^2} {\eps^4}} + {\frac {(k \log n)^{1/2}}{\eps^3}} \cdot \log \left({
\frac {k \log n}{\eps}} \right)\right)$ &
$\Omega\left(\left({k \log n}\right)^{1/2}\right)$
\\\hline

Testing identity, $q$ is unknown: &
$O\left({\frac {k^2} {\eps^4}} +
{\frac {\left(k \log n\right)^{2/3}} {\eps^{10/3}}} \cdot \log \left({\frac {k\log n}{\eps}}\right)\right)
$ & $\Omega\left(\left({\frac {k \log n}{\log (k\log n)}}\right)^{2/3}\right)$
\\\hline

Estimating $L_1$ distance, $q$ is known:  &
$O\left({\frac {k^2} {\eps^4}} + {\frac {k \log n}{\log(k \log n)}} \cdot \eps^{-4} \right)$  & $\Omega\left({\frac {k \log n} {\log(k \log n) \cdot \log \log (k\log n)}}\right)$ \\ \hline

Estimating $L_1$ distance, $q$ is unknown:&
 $O\left({\frac {k^2} {\eps^4}} + {\frac {k \log n}{\log(k \log n)}} \cdot  \eps^{-4}\right)$ & $\Omega\left({\frac {k \log n} {\log (k\log n) \cdot \log \log (k\log n)}}\right)$ \\\hline

\end{tabular}
\end{center}
\caption{Our upper and lower bounds for identity testing and $L_1$ estimation.
In the table we omit a ``$\log(1/\delta)$'' term which is present in all our upper bounds for
algorithms which give the correct answer with probability $1-\delta.$
For the ``testing identity'' problems, our lower bounds are for distinguishing whether
$p = q$ versus $\dtv(p,q) > 1/2$ with success probability $2/3.$  For estimating $L_1$ distance, our bounds are for estimating $\dtv(p,q)$ to within $\pm \eps$, for any $k=O(n^{1/2})$, with the lower bounds corresponding to success probability $2/3$.}
\label{tab:results}
\end{table}

The reduction-like techniques which we use to establish both our algorithmic results and our lower bounds
(discussed in more detail in Section~\ref{sec:techniques} below) reveal an unexpected relationship between the class of $k$-modal distributions of support $[n]$ and the class of general distributions of support $[k\log n]$. We hope that this reduction-based approach may provide a framework for the discovery of other
relationships that will be useful in future work in the extreme sublinear regime of statistical property estimation and property testing.

\medskip

\noindent {\bf Comparison with prior work.}
Our results significantly extend and improve upon the previous algorithmic results of Batu et al~\cite{BKR:04} for identity testing of monotone or unimodal ($k=1$) distributions, which required
$O(\log^3 n)$ samples.  More recently, \cite{DDS:11kmodallearn} established the sample complexity
of \emph{learning} $k$-modal distributions to be essentially $\Theta(k \log(n) \eps^{-3})$.
Such a learning algorithm easily yields a testing algorithm with the same sample complexity for all four
variants of the testing problem (one
can simply run the learner twice to obtain hypotheses $\wh{p}$
and $\wh{q}$ that are sufficiently close to $p$ and $q$
respectively, and output accordingly).

While the \cite{DDS:11kmodallearn} result can be applied to our testing problems (though giving
suboptimal results), we stress that the ideas underlying \cite{DDS:11kmodallearn} and this paper are quite
different.  The \cite{DDS:11kmodallearn} paper learns a $k$-modal distribution by using a known algorithm for learning monotone distributions \cite{Birge:87b} $k$ times in a black-box manner; 
the notion of \emph{reducing the domain size}---which we view as central to the results and contributions of this paper---is nowhere present in \cite{DDS:11kmodallearn}.
By contrast, the focus in this paper is on introducing the use of reductions as a powerful (but surprisingly, seemingly previously unused) tool in the development of algorithms for basic statistical tasks on distributions, which, at least in this case, is capable of giving essentially optimal upper
and lower bounds for natural restricted classes of distributions.

\subsection{Techniques.} \label{sec:techniques}

Our main conceptual contribution is a new reduction-based approach that lets us obtain all our upper and lower bounds in a clean and unified way.
The approach works by reducing the monotone and $k$-modal distribution testing problems to general distribution testing and estimation problems \emph{over a much smaller domain}, and vice versa.  For the monotone case this smaller domain is essentially of size $\log(n)/\eps$, and for the $k$-modal case the smaller domain is essentially of size $k \log(n)/\eps^2.$  By solving the general distribution problems
over the smaller domain using known results we get a valid answer for the original (monotone or $k$-modal) problems over domain $[n]$. More details on our algorithmic reduction are given in Section~\ref{sec:reduction}.

Conversely, our lower bound reduction lets us reexpress arbitrary distributions over a small domain $[\ell]$ by monotone (or unimodal, or $k$-modal, as required) distributions over an exponentially larger domain, while preserving many of their features with respect to the $L_1$ distance. Crucially, this reduction allows one to \emph{simulate} drawing samples from the larger monotone distribution given access to samples from the smaller distribution, so that a known impossibility result for unrestricted distributions on $[\ell]$ may be leveraged to yield a corresponding impossibility result for monotone (or unimodal, or $k$-modal) distributions on the exponentially larger domain.

The inspiration for our results is an observation of Birg\'e \cite{Birge:87b}\ignore{ (as reframed in \cite{DDS:11kmodallearn})} that given a monotone-decreasing probability distribution over $[n]$, if one subdivides $[n]$ into an exponentially increasing series of consecutive sub-intervals, the $i$th having size $(1+\epsilon)^i$, then if one replaces the probability mass on each interval with a uniform distribution on that interval, the distribution changes by only $O(\epsilon)$ in total variation distance. Further, given such a subdivision of the support into $\log_{1+\eps}(n)$ intervals, one may essentially treat the original monotone distribution as essentially a distribution over these intervals, namely a distribution of support $\log_{1+\eps}(n)$. In this way, one may hope to reduce monotone distribution testing or estimation on $[n]$ to general distribution testing or estimation on a domain of size $\log_{1+\eps}(n)$, and vice versa. See Section~\ref{sec:monotone} for details.

For the monotone  testing problems the partition into subintervals is constructed obliviously (without drawing any samples or making any reference
to $p$ or $q$ of any sort) -- for a given value of
$\eps$ the partition is the same for all non-increasing distributions.  For the $k$-modal testing problems, constructing the desired partition is significantly
more involved.  This is done via a careful procedure which uses $k^2 \cdot \poly(1/\eps)$ samples\footnote{Intuitively, the partition must be finer in regions of higher probability density; for non-increasing distributions (for example) this region is at the left side of the domain, but for general $k$-modal distributions, one must draw samples  to discover the high-probability regions.} from
$p$ and $q$ and uses the oblivious
decomposition for monotone distributions in a delicate way.  This construction is given in Section~\ref{sec:testkmodal}.

\section{Notation and Preliminaries} \label{sec:prelims}

\subsection{Notation.}
We write $[n]$ to denote the set $\{1, \ldots, n\}$, and for
integers $i\leq j$ we write  $[i,j]$ to denote the set $\{i, i+1, \ldots, j\}$.
We consider discrete probability distributions over $[n]$, which are
functions $p: [n] \to [0,1]$ such that $\littlesum_{i=1}^n p(i)=1$.
For $S \subseteq [n]$ we write $p(S)$ to denote $\littlesum_{i \in S} p(i)$.\ignore{, and we write
$p_S$ to denote the {\em conditional distribution}
over $S$ that is induced by $p,$ i.e. $p_S(i) = p(i)/p(S).$} We use the  notation $P$ for the
{\em cumulative distribution function (cdf)} corresponding to $p$, i.e.
$P: [n] \to [0,1]$ is defined by $P(j) = \littlesum_{i=1}^j p(i)$.

A distribution $p$ over $[n]$ is non-increasing (resp. non-decreasing)
if $p(i+1) \leq p(i)$ (resp. $p(i+1) \geq p(i)$), for all $i \in [n-1]$;
$p$ is \emph{monotone} if it is either non-increasing or non-decreasing.
Thus the ``orientation'' of a monotone distribution is either
non-decreasing (denoted $\uparrow$) or non-increasing
(denoted $\downarrow$).

We call a nonempty interval $I=[a,b] \subseteq [2, n-1]$ a \emph{max-interval}
of $p$ if $p(i) = c$ for all $i \in I$ and $\max\{p(a-1), p(b+1)\} < c$.
Analogously, a \emph{min-interval} of $p$ is an interval
$I=[a,b] \subseteq [2, n-1]$ with $p(i) = c$ for all $i \in I$
and $\min \{p(a-1), p(b+1)\} > c$.
We say that $p$ is \emph{$k$-modal} if it has at most $k$ max-intervals
and min-intervals.  We note that according to our definition, what is usually referred to as a bimodal distribution is a $3$-modal distribution.

Let $p, q$ be distributions over $[n]$ with corresponding cdfs $P, Q$.
The {\em total variation distance} between $p$ and $q$ is
$\dtv (p,q) := \max_{S \subseteq [n]} |p(S) - q(S)|  = (1/2) \littlesum_{i \in [n]}|p(i)-q(i)|.$
The {\em Kolmogorov distance} between $p$ and $q$ is defined as
$\dk(p,q):= \max_{j \in [n]} \left| P(j) - Q(j) \right|.$ Note that $\dk(p,q) \le \dtv (p,q).$

Finally, a \emph{sub-distribution} is a function $q: [n] \to [0,1]$ which satisfies
$\littlesum_{i=1}^n q(i) \leq 1.$  For $p$ a distribution over $[n]$ and $I \subseteq [n]$, the
\emph{restriction of $p$ to $I$} is the sub-distribution $p^I$ defined by
$p^I(i)=p(i)$ if $i \in I$ and $p^I(i)=0$ otherwise.
\ignore{
\[
p^I(i) = \begin{cases}
p(i) & \text{~if~}i \in I;\\
0 & \text{~if~}i \notin I.
\end{cases}
\]
} Likewise, we denote by $p_I$ the conditional distribution of $p$ on $I$,
i.e. \ignore{
\[
p_I(i) = \begin{cases}
p(i)/p(I) & \text{~if~}i \in I;\\
0 & \text{~if~}i \notin I.
\end{cases}
\]
}$p_I(i)=p(i)/p(I)$ if $i \in I$ and $p_I(i)=0$ otherwise.

\subsection{Basic tools from probability.}
We will require the \emph{ Dvoretzky-Kiefer-Wolfowitz
(DKW) inequality} (\cite{DKW56}) from probability theory.
This basic fact says that
$O(1/\eps^2)$ samples suffice to learn any distribution
within error $\eps$ with respect to the \emph{Kolmogorov distance}.
More precisely, let $p$ be any distribution over $[n].$
Given $m$ independent samples $s_1,\dots,s_m$
drawn from $p:[n] \to [0,1],$
the {\em empirical distribution} $\wh{p}_m : [n] \to [0,1]$ is defined as
follows: for all $i \in [n]$, $\wh{p}_m(i) = |\{j \in [m] \mid s_j=i\}| / m$.
The DKW inequality states that for $m=\Omega((1/\eps^2)\cdot \ln(1/\delta))$,
with probability $1-\delta$ the {empirical distribution} $\wh{p}_m$ will be
$\eps$-close to $p$ in Kolmogorov distance.
This sample bound is asymptotically optimal and independent of the support size.

\begin{theorem}[\cite{DKW56,Massart90}] \label{thm:DKW}
For all $\eps>0$, it holds: $\Pr [ \dk(p, \wh{p}_m) > \eps ] \leq 2e^{-2m\eps^2}.$
\end{theorem}

Another simple result that we will need is the following, which is easily verified
from first principles:

\begin{observation} \label{obs:mon}
Let $I=[a,b]$ be an interval and let $u_I$ denote the uniform
distribution over $I.$  Let $p_I$ denote a non-increasing distribution over $I$.
Then for every initial interval $I' = [a,b']$ of $I$, we have
$u_I(I') \leq p_I(I').$
\end{observation}


\subsection{Testing and estimation for arbitrary distribution}
Our testing algorithms work by reducing to known algorithms for testing arbitrary distributions
over an $\ell$-element domain.  We will use the following well known results: \ignore{\rnote{Somebody please check (i) all the bounds claimed below, and (ii) that what's in our table in the Introduction matches these
bounds.  If someone does this, remove this comment.}}

\begin{theorem}  [testing identity, known distribution \cite{BFFKRW:01}] \label{thm:testidentityknown}
Let $q$ be an explicitly given distribution over $[\ell]$.  Let $p$ be
an unknown distribution over $[\ell]$ that is accessible via samples.
There is a testing algorithm
\textsc{{Test-Identity-Known}}$(p, q, \eps,\delta)$
that uses $s_{IK}(\ell,\eps,\delta) := O(\ell^{1/2} \log(\ell) \eps^{-2} \log(1/\delta))$ samples from $p$ and has the following properties:

\begin{itemize}
\vspace{-.2cm} \item If $p \equiv q$ then with probability at least
$1-\delta$ the algorithm outputs ``accept;'' and
\vspace{-.2cm} \item If $\dtv(p,q) \geq \eps$ then with probability at least $1-\delta$
the algorithm outputs ``reject.''
\end{itemize}
\end{theorem}

\begin{theorem}
[testing identity, unknown distribution \cite{BFRSW:10}]\label{thm:testidentityunknown}
Let $p$ and $q$ both be unknown distributions over $[\ell]$ that are accessible via samples.
There is a testing algorithm \textsc{ Test-Identity-Unknown}$(p, q, \eps,\delta)$
that uses $s_{IU}(\ell,\eps,\delta) := O(\ell^{2/3} \log(\ell/\delta) \eps^{-8/3})$ samples from $p$ and $q$ and has the following properties:

\begin{itemize}
\vspace{-.2cm} \item If $p \equiv q$ then with probability at least
$1-\delta$ the algorithm outputs ``accept;'' and

\vspace{-.2cm} \item If $\dtv(p,q) \geq \eps$ then with probability at least $1-\delta$
the algorithm outputs ``reject.''
\end{itemize}
\end{theorem}

\begin{theorem}  [$L_1$ estimation \cite{ValiantValiant:11focs}] \label{thm:testtolerantknown}
  Let $p$ be
an unknown distribution over $[\ell]$ that is accessible via samples, and let $q$ be a distribution over $[\ell]$ that is either explicitly given, or accessible via samples.
There is an estimator
\textsc{ $L_1$-Estimate}$(p, q, \eps, \delta)$ that, with probability at least $1-\delta$, outputs a value in the interval $(\dtv(p,q)-\eps, \dtv(p,q)+\eps).$
The algorithm uses $s_{E}(\ell,\eps,\delta) := O\left({\frac \ell {\log \ell}} \cdot \eps^{-2}\log(1/\delta)\right)$ samples.
\end{theorem}


\section{Testing and Estimating Monotone Distributions}

\subsection{Oblivious decomposition of monotone distributions} \label{sec:obliv}
Our main tool for testing monotone distributions is an {\em oblivious decomposition} of monotone distributions that is a variant of a construction of Birg\'{e}
\cite{Birge:87b}. As we will see it enables us to reduce the problem of
testing a monotone distribution to the problem of testing an arbitrary distribution over a much
smaller domain.

Before stating the decomposition, some notation will be helpful.  Fix a distribution $p$ over $[n]$ and a partition of $[n]$ into disjoint intervals
$\mathcal{I} :=  \{ I_i \}_{i=1}^{\ell}.$
The {\em flattened distribution} $(p_f)^{\mathcal{I}}$ corresponding to $p$ and $\mathcal{I}$
is the distribution over $[n]$ defined as follows:  for $j \in [\ell]$ and $i \in I_j$, $(p_f)^{\mathcal{I}} (i) = \littlesum_{t \in I_j} p(t) / |I_j|$.
That is, $(p_f)^{\mathcal{I}}$ is obtained from $p$ by averaging the weight that $p$ assigns to each interval over the entire interval.
The {\em reduced distribution} $(p_r)^{\mathcal{I}}$ corresponding to $p$ and $\mathcal{I}$ is the distribution over $[\ell]$
that assigns the $i$th point the weight $p$ assigns to the
interval $I_i$; i.e., for $i \in [\ell]$, we have $(p_r)^{\mathcal{I}} (i) = p(I_i)$.
Note that if $p$ is non-increasing then so is $(p_f)^{\mathcal{I}}$, but
this is not necessarily the case for $(p_r)^{\mathcal{I}}$.

The following simple lemma, proved in Section~\ref{sec:reduction}, shows
why reduced distributions are useful for us:

\begin{definition} \label{def:flatdecomp}
Let $p$ be a distribution over $[n]$ and let ${\cal I}= \{ I_i \}_{i=1}^{\ell}$ be a partition
of $[n]$ into disjoint intervals.  We say that ${\cal I}$ is a \emph{$(p,\eps,\ell)$-flat decomposition of $[n]$} if
$\dtv(p,(p_f)^{{\cal I}}) \leq \eps.$
\end{definition}

\begin{lemma} \label{lem:reduction}
Let ${\cal I} = \{ I_i \}_{i=1}^{\ell}$ be a partition of $[n]$ into disjoint intervals.
Suppose that $p$ and $q$ are distributions over $[n]$ such that ${\cal I}$ is both a $(p,\eps,\ell)$-flat decomposition of $[n]$
and is also a $(q,\eps,\ell)$-flat decomposition of $[n]$.  Then
$ \left| \dtv(p,q) - \dtv ( (p_r)^{\cal I}, (q_r)^{\cal I} ) \right|  \leq 2\eps.$
Moreover, if $p = q$ then $(p_r)^{\cal I} = (q_r)^{\cal I}$.
\end{lemma}

We now state our oblivious decomposition result for monotone distributions:

\begin{theorem}[\cite{Birge:87b}] \label{thm:birge-oblivious}
(oblivious decomposition)  Fix any $n \in \mathbb{Z}^+$ and $\eps>0$.
The partition $\mathcal{I} : = \{ I_i\}_{i=1}^{\ell}$ of $[n],$ in which the $j$th interval has size $\lfloor (1+ \eps)^j \rfloor$ has the following properties:  $\ell = O\left( (1/\eps) \cdot \log (\eps \cdot n + 1) \right)$,
and  ${\cal I}$ is a $(p,O(\eps),\ell)$-flat decomposition of $[n]$
for any non-increasing distribution $p$ over
$[n]$.
\end{theorem}

There is an analogous version of Theorem~\ref{thm:birge-oblivious}, asserting the existence of an ``oblivious''
partition for non-decreasing distributions (which is of course different from the ``oblivious''
partition ${\cal I}$ for non-increasing distributions of Theorem~\ref{thm:birge-oblivious}); this will be useful later.


While our construction is essentially that of Birg\'{e}, we
note that the version given in \cite{Birge:87b} is for
non-increasing distributions over the continuous domain $[0,n]$,
and it is phrased rather differently.  Adapting the arguments
of \cite{Birge:87b} to our discrete setting of
distributions over $[n]$ is not conceptually difficult but requires
some care.  For the sake of being self-contained we
provide a self-contained proof of the discrete version, stated above, that
we require in Appendix~\ref{ap:birge}.

\subsection{Efficiently testing monotone distributions} \label{sec:testmonotone}

Now we are ready to establish our upper bounds on testing monotone distributions
(given in the first four rows of Table~1).
All of the algorithms are essentially the same: each works by reducing the given monotone distribution testing problem to the same testing problem for arbitrary distributions over support of size $\ell = O(\log n / \eps)$
using the oblivious decomposition from the previous subsection.
For concreteness we explicitly describe the tester for the ``testing identity, $q$ is known'' case below, and then indicate the small changes that are necessary to get the testers for the other three cases.

\medskip

\ignore{\hskip-.2in}
\noindent\framebox{

\medskip \noindent \begin{minipage}{17cm}

\medskip

\textsc{ Test-Identity-Known-Monotone}

\noindent {\bf Inputs:} $\eps,\delta >0$; sample access to non-increasing
distribution $p$ over $[n]$; explicit description of non-increasing distribution
$q$ over $[n]$

\begin{enumerate}
\vspace{-.2cm}
\item Let $\mathcal{I} : = \{I_i\}_{i=1}^{\ell}$, with $\ell= \Theta (\log(\eps n+1)/\eps)$, be the partition of $[n]$ given by Theorem~\ref{thm:birge-oblivious}, which is a $(p',\eps/8,\ell)$-flat decomposition of $[n]$ for any non-increasing distribution $p'$.

\vspace{-.2cm}
\item Let $(q_r)^{\cal I}$ denote the reduced distribution over $[\ell]$ obtained from $q$ using
${\cal I}$ as defined in Section~\ref{sec:reduction}.
\vspace{-.2cm}
\item  Draw $m = s_{IK}(\ell,\eps/2,\delta)$  samples 
        from $(p_r)^{\cal I}$, where $(p_r)^{\cal I}$ is the reduced distribution over $[\ell]$ obtained from $p$ using ${\cal I}$ as defined in Section~\ref{sec:reduction}.

\vspace{-.2cm}
\item Output the result of \textsc{ Test-Identity-Known}$((p_r)^{\cal I}, (q_r)^{\cal I}, {\frac \eps 2},\delta)$ on the samples from Step~3.

\end{enumerate}
\end{minipage}}
\medskip

We now establish our claimed upper bound for the ``testing identity, $q$ is known'' case.
%
%
We first observe that in Step~3, the desired $m=s_{IK}(\ell,\eps/2,\delta)$ samples from $(p_r)^{\cal I}$ can easily be obtained by drawing $m$ samples from $p$ and converting each one to the corresponding draw from
$(p_r)^{\cal I}$ in the obvious way.
If $p = q$ then $(p_r)^{\cal I} = (q_r)^{\cal I}$, and \textsc{ Test-Identity-Known-Monotone} outputs ``accept'' with probability at least $1 - \delta$ by Theorem~\ref{thm:testidentityknown}.
If $\dtv(p,q) \geq \eps$, then by Lemma~\ref{lem:reduction}, Theorem~\ref{thm:birge-oblivious} and the triangle inequality, we have that $\dtv((p_r)^{\cal I},(q_r)^{\cal I}) \geq 3\eps/4$, so \textsc{ Test-Identity-Known-Monotone} outputs ``reject'' with probability at least $1 - \delta$ by Theorem~\ref{thm:testidentityknown}.
For the ``testing identity, $q$ is unknown'' case, the the algorithm \textsc{ Test-Identity-Unknown-Monotone} is very similar to \textsc{ Test-Identity-Known-Monotone}.  The differences are as follows:  instead of Step~2, in Step~3 we draw
$m = s_{IU}(\ell,\eps/2,\delta)$ samples from $(p_r)^{\cal I}$ and the same number
of samples from $(q_r)^{\cal I}$; and in Step~4, we run \textsc{ Test-Identity-Unknown}$((p_r)^{\cal I},(q_r)^{\cal I},
{\frac \eps 2},\delta)$ using the samples from Step~3.
The analysis is exactly the same as above (using Theorem~\ref{thm:testidentityunknown} in place
of Theorem~\ref{thm:testidentityknown}).

We now describe the algorithm \textsc{$L_1$-Estimate-Known-Monotone} for the ``tolerant testing, $q$ is known'' case. This algorithm takes values $\eps$ and $\delta$ as input, so the partition ${\cal I}$ defined in Step~1 is a $(p',\eps/4,\ell)$-flat decomposition of $[n]$ for any non-increasing $p'.$
In Step~3 the algorithm draws $m = s_{E}(\ell,\eps/2,\delta)$ samples
and runs \textsc{ $L_1$-Estimate}$((p_r)^{\cal I}, (q_r)^{\cal I},\eps/2,\delta)$ in Step~4.
If $\dtv(p,q)=c$ then by the triangle inequality we have that $\dtv((p_r)^{\cal I},(q_r)^{\cal I}) \in [c-\eps/2,c+\eps/2]$ and \textsc{ $L_1$-Estimate-Known-Monotone} outputs a value within the prescribed range with probability at least $1-\delta,$ by Theorem~\ref{thm:testtolerantknown}.
The algorithm \textsc{$L_1$-Estimate-Unknown-Monotone}
and its analysis are entirely similar.

\section{From Monotone to $k$-modal}

In this section we establish our main positive testing results for $k$-modal distributions, the upper bounds
stated in the final four rows of Table~1.  In the previous section, we were able to use the oblivious decomposition to yield a partition of $[n]$ into relatively few intervals, with the guarantee that the corresponding flattened distribution is close to the true distribution.  The main challenge in extending these results to unimodal or $k$-modal distributions, is that  in order to make the  analogous decomposition, one must first determine--by taking samples from the distribution--which regions are monotonically increasing vs decreasing.  Our algorithm \textsc{ Construct-Flat-Decomposition}$(p,\eps,\delta)$ performs this task with the following guarantee:

\begin{lemma} \label{lem:cfsd}
Let $p$ be a $k$-modal distribution over $[n]$.  Algorithm \textsc{ Construct-Flat-Decomposition}$(p,\eps,\delta)$
draws $O(k^2 \eps^{-4} \log(1/\delta))$ samples from $p$ and outputs a $(p,\eps,\ell)$-flat decomposition of $[n]$ with probability at least $1-\delta$, where $\ell = O(k \log(n)/\eps^2)$.
\end{lemma}

The bulk of our work in Section~\ref{sec:testkmodal} is to describe \textsc{ Construct-Flat-Decomposition}$(p,\eps,\delta)$ and prove Lemma~\ref{lem:cfsd}, but first we show how Lemma~\ref{lem:cfsd} yields our claimed testing results for $k$-modal distributions.
As in the monotone case all four algorithms are essentially the same:  each works by reducing the given $k$-modal distribution testing problem to the same testing problem for arbitrary distributions over $[\ell].$ One slight complication is that the partition obtained for distribution $p$ will generally differ from that for $q$.  In the monotone distribution setting, the partition was oblivious to the distributions, and thus this concern did not arise.  Naively, one might hope that the flattened distribution corresponding to any refinement of a partition will be at least as good as the flattened distribution corresponding to the actual partition.  This hope is easily seen to be strictly false, but we show that it is true up to a factor of 2, which suffices for our purposes.

The following terminology will be useful:  Let ${\cal I}= \{I_i\}_{i=1}^r$ and ${\cal I}'=\{I'_i\}_{i=1}^s$ be two
partitions of $[n]$ into $r$ and $s$ intervals respectively.  The \emph{common refinement} of ${\cal I}$
 and ${\cal I}'$ is the partition ${\cal J}$ of $[n]$ into intervals obtained from ${\cal I}$ and ${\cal I}'$ in the obvious way, by taking all possible nonempty intervals of the form
 $I_i \cap I'_j.$  It is clear that ${\cal J}$ is both a refinement of ${\cal I}$ and of ${\cal I}'$
 and that the number of intervals $|{\cal J}|$ in ${\cal J}$ is at most $r+s.$  We prove the following lemma
 in Section~\ref{sec:reduction}:

\begin{lemma} \label{lem:refinement}
Let $p$ be any distribution over $[n]$, let
${\cal I}=\{I_i\}_{i=1}^{a}$ be a $(p,\eps,a)$-flat
decomposition of $[n]$, and let ${\cal J}=\{J_i\}_{i=1}^b$ be
a refinement of ${\cal I}.$ Then ${\cal J}$ is a $(p,2\eps,b)$-flat
decomposition of $[n]$.
\end{lemma}

 We describe the \textsc{ Test-Identity-Known-kmodal} algorithm below. \ignore{and then indicate the necessary changes to get the other three testers.}


\ignore{(Note that it uses samples from $q$
even though it has an explicit description of $q$; this is for consistency with the
later \textsc{ Test-kmodal-Unknown} algorithm, and does not hurt the sample
complexity of \textsc{ Test-kmodal-Known}.)}

\medskip

\ignore{\hskip-.2in}

\noindent \framebox{

\medskip \noindent \begin{minipage}{17cm}

\medskip

\textsc{ Test-Identity-Known-kmodal}

\noindent {\bf Inputs:} $\eps,\delta >0$; sample access to $k$-modal
distributions $p,q$ over $[n]$

\begin{enumerate}

\vspace{-.2cm} \item Run \textsc{ Construct-Flat-Decomposition}$(p,\eps/2,\delta/4)$ and let ${\cal I}=\{I_i\}_{i=1}^{\ell},$
    $\ell = O(k\log(n)/\eps^2)$, be the partition that it outputs.  Run
    \textsc{ Construct-Flat-Decomposition}$(p,\eps/2,\delta/4)$ and let ${\cal I}'=\{I'_i\}_{i=1}^{\ell'},$
    $\ell' = O(k\log(n)/\eps^2)$, be the partition that it outputs.  Let ${\cal J}$ be the
    common refinement of ${\cal I}$ and ${\cal I}$ and
    let $\ell_{\cal J} = O(k \log(n)/\eps^2)$ be the number of intervals in ${\cal J}.$

\vspace{-.2cm} \item Let $(q_r)^{\cal J}$ denote the reduced distribution over $[\ell_{\cal J}]$ obtained from $q$
using ${\cal J}$ as defined in Section~\ref{sec:reduction}.

\vspace{-.2cm} \item  Draw $m = s_{IK}(\ell_{\cal J},\eps/2,\delta/2)$ samples 
        from $(p_r)^{\cal J}$, where $(p_r)^{\cal J}$ is the reduced distribution over $[\ell_{\cal J}]$ obtained from $p$ using ${\cal J}$ as defined in Section~\ref{sec:reduction}.


\vspace{-.2cm} \item Run \textsc{ Test-Identity-Known}$((p_r)^{\cal J}, (q_r)^{\cal J}, {\frac \eps 2},{\frac \delta 2})$ using the samples from Step~3 and output what it outputs.


\end{enumerate}
\end{minipage}}
\medskip

We note that Steps~2, 3 and~4 of \textsc{ Test-Identity-Known-kmodal} are the same as the corresponding
steps of \textsc{ Test-Identity-Known-Monotone}.  For the analysis of \textsc{ Test-Identity-Known-kmodal},
 Lemmas~\ref{lem:cfsd} and~\ref{lem:refinement} give us that with probability $1 - \delta/2$, the partition
${\cal J}$ obtained in Step~1 is both a $(p,\eps,\ell_{\cal J})$-flat and $(q,\eps,\ell_{\cal J})$-flat decomposition of $[n]$;
we condition on this going forward.  From this point on the analysis is essentially identical to
the analysis for \textsc{ Test-Identity-Known-Monotone}
and is omitted.

The modifications required to obtain
algorithms \textsc{ Test-Identity-Unknown-kmodal}, \textsc{ $L_1$-Estimate-Known-kmodal} and
\textsc{ $L_1$-Estimate-Unknown-kmodal}, and the analysis of these algorithms, are completely analogous to the
modifications and analyses of Section~\ref{sec:testmonotone} and are omitted.

\subsection{The \textsc{ Construct-Flat-Decomposition} algorithm.}
We present \textsc{ Construct-Flat-Decomposition}$(p,\eps,\delta)$ followed by an intuitive explanation.  Note that it employs a procedure \textsc{ Orientation}$(\wh{p},I)$, which uses no samples and is presented and analyzed in Section~\ref{sec:kmodaltools}.

\medskip

\noindent\framebox{

\medskip \noindent \begin{minipage}{17cm}

\medskip

\textsc{ Construct-Flat-Decomposition}

\noindent \textsc{ Inputs:} $\eps,\delta >0$; sample access to $k$-modal
distribution $p$ over $[n]$

\begin{enumerate}

\vspace{-.2cm} \item Initialize ${\cal I} := \emptyset.$

\vspace{-.2cm} \item Fix $\tau:=\eps^2/(20000k)$. Draw $ r = \Theta(\log(1/\delta)/\tau^2)$ samples
from $p$ and let $\wh{p}$ denote the resulting empirical distribution (which by Theorem~\ref{thm:DKW} has
$\dk(\wh{p},p) \leq \tau$ with probability at least $1-\delta$).

\vspace{-.2cm} \item Greedily partition the domain $[n]$ into $\alpha$ {\em atomic intervals} $\{I_i\}_{i=1}^{\alpha}$ as follows:
$I_1 := [1, j_1]$, where $j_1 := \min\{ j \in [n] \mid \wh{p}([1,j]) \geq \eps/(100k)\}$.
For $i \geq 1$, if $\cup_{j=1}^i I_j = [1, j_i]$, then
$I_{i+1}:=[j_i+1, j_{i+1}]$, where $j_{i+1}$ is defined as follows:
If $\wh{p}([j_i+1, n]) \geq \eps/(100k)$, then $j_{i+1}: = \min\{ j \in [n]
\mid \wh{p}([j_{i}+1, j]) \geq \eps/(100k)\}$, otherwise, $j_{i+1} := n$.

\vspace{-.2cm} \item Construct a set of $n_m$ {\em moderate intervals}, a set of
$n_h$ {\em heavy points}, and a set of $n_n$ {\em
negligible intervals} as follows:  For each atomic interval $I_i = [a,b] $,
\begin{itemize}

\vspace{-.2cm} \item [(a)] if $\wh{p}([a,b]) \leq 3\eps/(100k)$ then $I_i$ is declared to be
a \emph{moderate} interval;

\vspace{-.2cm} \item [(b)] otherwise we have $\wh{p}([a,b]) > 3\eps/(100k)$ and we declare
$b$ to be a \emph{heavy point}.  If $a < b$ then we declare $[a,b-1]$
to be a \emph{negligible} interval.
\end{itemize}

For each interval $I$ which is a heavy point, add $I$ to ${\cal I}.$
Add each negligible interval $I$  to ${\cal I}.$

\vspace{-.2cm} \item For each moderate interval $I$, run procedure \textsc{ Orientation}$(\wh{p},I)$; let $\circ \in \{\uparrow,\downarrow,\bot\}$ be its output.

     If $\circ = \bot$ then add $I$  to ${\cal I}.$

    If $\circ = \downarrow$ then let ${\cal J}_I$ be the partition of $I$ given
    by Theorem~\ref{thm:birge-oblivious} which is a $(p',\eps/4,O(\log(n)/\eps))$-flat decomposition of
    $I$ for any non-increasing distribution $p'$ over $I.$  Add all the elements of ${\cal J}_I$ to ${\cal I}.$

    If $\circ = \uparrow$ then let ${\cal J}_I$ be the partition of $I$ given
    by the dual version of Theorem~\ref{thm:birge-oblivious}, which is a $(p',\eps/4,O(\log(n)/\eps))$-flat decomposition of $I$ for any non-decreasing distribution $p'$ over $I.$ Add all the elements of ${\cal J}_I$ to ${\cal I}.$

\vspace{-.2cm} \item Output the partition ${\cal I}$ of $[n]$.

\end{enumerate}
\end{minipage}}

\medskip

Roughly speaking, when \textsc{ Construct-Flat-Decomposition} constructs a partition ${\cal I}$, it initially breaks $[n]$ up into
two types of intervals.  The first type are intervals that are ``okay'' to include in
a flat decomposition, either because they have very little mass, or because they
consist of a single point, or because they are close to uniform.  The second type are intervals
that are ``not okay'' to include in a flat decomposition -- they have significant mass
and are far from uniform -- but the algorithm is able to ensure that almost all of these are monotone
distributions with a known orientation.  It then uses the oblivious decomposition of
Theorem~\ref{thm:birge-oblivious} to construct a flat decomposition of each such interval.  (Note that
it is crucial that the orientation is known in order to be able to use Theorem~\ref{thm:birge-oblivious}.)

In more detail, \textsc{ Construct-Flat-Decomposition}$(p,\eps,\delta)$
works  as follows.  The algorithm first draws a batch of samples from $p$ and uses them to construct an
estimate $\wh{p}$ of the CDF of $p$ (this is straightforward using the DKW
inequality). Using $\wh{p}$ the algorithm partitions $[n]$ into a collection
of $O(k/\eps)$ disjoint intervals in the following way:
\begin{itemize}
\vspace{-.2cm} \item A small collection of the intervals are ``negligible''; they collectively have total mass less than $\eps$ under $p$.  Each negligible interval $I$ will be an element of the partition ${\cal I}.$

\vspace{-.2cm} \item Some of the intervals are ``heavy points''; these are intervals consisting of a single point
that has mass $\Omega(\eps/k)$ under $p$.  Each heavy point $I$ will also be an element of the
partition ${\cal I}.$

\vspace{-.2cm} \item The remaining intervals are ``moderate'' intervals, each of which has mass $\Theta(\eps/k)$ under $p$.
\end{itemize}

It remains to incorporate the moderate intervals into the partition ${\cal I}$ that is being constructed.
This is done as follows:  using $\wh{p}$, the algorithm comes up with a ``guess'' of the correct orientation (non-increasing, non-decreasing, or close to uniform) for each moderate interval.  Each moderate interval where the ``guessed'' orientation is ``close to uniform'' is included in the partition ${\cal I}.$ Finally, for each moderate interval $I$ where the guessed orientation is ``non-increasing'' or ``non-decreasing'', the algorithm invokes Theorem~\ref{thm:birge-oblivious} on $I$ to perform the oblivious decomposition for monotone distributions, and the resulting sub-intervals are included in ${\cal I}$.
The analysis will show that the guesses are almost always
correct, and intuitively this should imply that  the ${\cal I}$ that is constructed is indeed a $(p,\eps,\ell)$-flat decomposition of $[n]$.

\subsection{The \textsc{ Orientation} algorithm.} \label{sec:kmodaltools}

The \textsc{ Orientation} algorithm takes as input an explicit distribution of a distribution $\wh{p}$ over
$[n]$ and an interval $I \subseteq [n].$  Intuitively, it assumes that $\wh{p}_I$ is close
(in Kolmogorov distance) to a monotone distribution $p_I$, and its goal is to determine the orientation of $p_I$: it outputs either $\uparrow$, $\downarrow$ or $\bot$ (the last of which means ``close to uniform'').
The algorithm is quite simple; it checks whether there exists an
initial interval $I'$ of $I$ on which $\wh{p}_I$'s weight is significantly different from $u_I(I')$ (the weight
that the uniform distribution over $I$ assigns to $I'$) and bases its output on this in the obvious way.  A precise description of the algorithm (which uses no samples) is given below.

\bigskip

\noindent \framebox{

\medskip \noindent \begin{minipage}{17cm}

\medskip

\textsc{ Orientation}

\noindent \textsc{ Inputs:} explicit description of distribution $\wh{p}$ over $[n]$; interval
$I =[a,b] \subseteq [n]$

\smallskip

\begin{enumerate}

\vspace{-.2cm} \item If $|I|=1$ (i.e. $I = \{a\}$ for some $a \in [n]$) then return ``$\bot$'', otherwise continue.

\vspace{-.2cm} \item If there is an initial interval $I'=[a,j]$ of $I$ that satisfies
$
u_I(I')-(\wh{p})_I(I') > {\frac \eps 7}
$
then halt and output ``$\uparrow$''.  Otherwise,

\vspace{-.2cm} \item If there is an initial interval $I'=[a,j]$ of $I$ that satisfies
$
u_I(I')-(\wh{p})_I(I') < - {\frac \eps 7}$
then halt and output ``$\downarrow$''.  Otherwise,

\vspace{-.2cm} \item Output ``$\bot$''.

\end{enumerate}

\end{minipage}}
\medskip

We proceed to analyze \textsc{ Orientation}. We show that if $p_I$ is far from uniform then \textsc{ Orientation} outputs the correct orientation for it. We also show that whenever \textsc{ Orientation} does not output ``$\bot$'', whatever it outputs is the correct orientation of $p_I$. The proof is given in Section~\ref{sec:kmodaltools2}.

\begin{lemma} \label{claim:orientation}
Let $p$ be a distribution over $[n]$ and let interval $I =[a,b] \subseteq [n]$ be such that $p_I$ is monotone.
Suppose $p(I) \geq 99\eps/(10000k)$, and suppose that for every interval $I' \subseteq I$ we have that $|\wh{p}(I')-p(I')| \leq {\frac {\eps^2}{10000k}}.$
Then
\begin{enumerate}
\vspace{-.2cm} \item If $p_I$ is non-decreasing and $p_I$ is $\eps/6$-far from the uniform distribution
$u_I$ over $I$, then \textsc{ Orientation}$(\wh{p},I)$ outputs ``$\uparrow$'';

\vspace{-.2cm} \item if \textsc{ Orientation}$(\wh{p},I)$ outputs ``$\uparrow$'' then $p_I$ is non-decreasing;

\vspace{-.2cm} \item if $p_I$ is non-increasing and $p_I$ is $\eps/6$-far from the uniform distribution
$u_I$ over $I$, then \textsc{ Orientation}$(\wh{p},I)$ outputs ``$\downarrow$'';

\vspace{-.2cm} \item if \textsc{ Orientation}$(\wh{p},I)$ outputs ``$\downarrow$'' then $p_I$ is non-increasing.
\end{enumerate}
\end{lemma}

\section{Lower Bounds}

Our algorithmic results follow from a reduction which shows how one can reduce the problem of testing properties of monotone or $k$-modal distributions to the task of testing properties of general distributions over a much smaller support.  Our approach to proving lower bounds is complementary; we give a canonical scheme for transforming ``lower bound instances'' of general distributions to related lower bound instances of monotone distributions with much larger supports.

A generic lower bound instance for distance estimation has the following form: there is a distribution $D$ over  \emph{pairs} of distributions, $(p,p')$, with the information theoretic guarantee that, given $s$ independent samples from distributions $p$ and $p'$, with $(p,p')\leftarrow D,$ it is impossible to distinguish the case that $\dtv(p,p') \le \eps_1$ versus $\dtv(p,p') > \eps_2$ with any probability greater than $1-\delta$, where the probability is taken over both the selection of $(p,p')\leftarrow D$ and the choice of samples.  In general, such information theoretic lower bounds are difficult to prove.  Fortunately, as mentioned above, we will be able to prove lower bounds for monotone and $k$-modal distributions by leveraging the known lower bound constructions in a black-box fashion.

Definitions~\ref{def:lb} and ~\ref{def:lb2}, given below, define a two-stage transformation of a generic distribution into a related $k$-modal distribution over a much larger support.  This transformation preserves total variation distance:  for any pair of distributions, the variation distance between their transformations is identical to the variation distance between the original distributions.  Additionally, we ensure that given access to $s$ independent samples from an original input distribution, one can \emph{simulate} drawing $s$ samples from the related $k$-modal distribution yielded by the transformation.  Given any lower--bound construction $D$ for general distributions, the above transformation will yield a lower--bound instance $D_k$ for $(k-1)$-modal distributions (so monotone distributions correspond to $k=1$) defined by selecting a pair of distributions $(p,p')\leftarrow D,$ then outputting the pair of transformed distributions.  This transformed ensemble of distributions is a lower--bound instance, for if some algorithm could successfully test pairs of $(k-1)$-modal distributions from $D_k,$ then that algorithm could be used to test pairs from $D$, by \emph{simulating} samples drawn from the transformed versions of the distributions.
The following proposition, proved in Section~\ref{sec:proplb}, summarizes the above discussion:

\begin{prop}\label{prop:lb}
Let $D$ be a distribution over pairs of distributions supported on $[n]$ such that given $s$ samples from distributions $p,p'$ with $(p,p') \leftarrow D,$ no algorithm can distinguish whether $\dtv(p,p')\le \eps_1$ versus $\dtv(p,p') > \eps_2$ with probability greater than $1-\delta$ (over both the draw of $(p,p')$ from
$D$ and the draw of samples from $p,p'$).  Let $p_{max},p_{min}$ be the respective maximum and minimum probabilities with which any element arises in distributions that are supported in $D$.    Then there exists a distribution $D_k$ over pairs of $(k-1)$-modal distributions supported on $[N] = [4ke^{\frac{8n}{k}\left(1+ \log(p_{max}/p_{min}) \right)}]$  such that no algorithm, when given $s$ samples from distributions $p_k,p'_k,$ with $(p_k,p'_k) \leftarrow D_k$,  can distinguish whether $\dtv(p_k,p'_k)\le \eps_1$ versus $\dtv(p_k, p'_k) > \eps_2$ with success probability greater than $1-\delta$. \end{prop}

Before proving this proposition, we state various corollaries which result from applying the Proposition to known lower-bound constructions for general distributions.  The first is for the ``testing identity, $q$ is
unknown'' problem:

\begin{corollary}
There exists a constant $c$ such that for sufficiently large $N$ and $1 \leq k=O(\log N)$,  there is a distribution $D_k$ over pairs of $2(k-1)$-modal distributions $(p,p')$ over $[N]$, such that \emph{no} algorithm, when given $c \left(\frac{k \log N}{\log \log N} \right)^{2/3}$ samples from a pair of distributions $(p,p') \leftarrow D,$ can distinguish the case that $\dtv(p,p')=0$ from the case $\dtv(p,p')> .5$ with probability at least $.6$.
\end{corollary}

This Corollary gives the lower bounds stated in lines~2 and~6 of Table~1.  It follows from applying Proposition~\ref{prop:lb} to a (trivially modified) version of the lower bound construction given in~\cite{BFR+:00,Valiant08stoc}, summarized by the following theorem:

\begin{theorem}[\cite{BFR+:00,Valiant08stoc}]
There exists a constant $c$ such that for sufficiently large $n$, there is a distribution $D$ over pairs of distributions $(p,p')$ over $[n]$, such that for any $(p,p') \leftarrow D,$ the maximum probability with which any element occurs in $p$ or $p'$ is $\frac{1}{n^{2/3}},$ and the minimum probability is $\frac{1}{2n}$.  Additionally, no algorithm, when given $c n^{2/3}$ samples from $(p,p')\leftarrow D,$ can distinguish whether $\dtv(p,p')=0,$ from $\dtv(p,p')>.5$ with probability at least $.6$.\end{theorem}

Our second corollary is for $L_1$ estimation, in the case that one of the distributions is explicitly given.  This trivially also yields an equivalent lower bound for the setting in which both distributions are given via samples.

\begin{corollary}
For any $a,b$ with $0<a<b<1/2,$ there exists a constant $c>0$, such that for any sufficiently large $N$ and $1 \leq k=O(\log N)$, there exists a $2(k-1)$-modal distribution $q$ of support $[N]$, and a distribution $D_k$ over $2(k-1)$-modal distributions over $[N]$, such that \emph{no} algorithm, when given $c \frac{k \log N}{\log \log N \cdot \log \log \log N}$ samples from a distribution $p \leftarrow D,$ can distinguish the case that $\dtv(p,q)<a$ versus $\dtv(p,p')> b$ with probability at least $.6$.
\end{corollary}

This Corollary gives the lower bounds claimed in lines~3,~4,~7 and~8 of Table~1.  It follows from applying Proposition~\ref{prop:lb} to the lower bound construction given in~\cite{ValiantValiant:11}, summarized by the following theorem:
\begin{theorem}[\cite{ValiantValiant:11}]~\label{thm:vv11f}
For any $a,b$ with $0<a<b<1/2,$ there exists a constant $c>0$, such that for any sufficiently large $n$, there is a distribution $D$ over distributions with support $[n]$, such that for any $p \leftarrow D,$ the maximum probability with which any element occurs in $p$ is $O\left(\frac{\log n}{n}\right),$ and the minimum probability is $\frac{1}{2n}$.  Additionally, no algorithm, when given $c\frac{n}{\log n}$ samples from $p \leftarrow D$ can distinguish whether $\dtv(p,u_n)<a$ versus $\dtv(p,u_n)>b$ with probability at least $.6$, where $u_n$ denotes the uniform distribution over $[n]$.
\end{theorem}

Note that the above theorem can be expressed in the language of Proposition~\ref{prop:lb} by defining the distribution $D'$ over pairs of distributions which chooses a distribution according to $D$ for the first distribution of each pair, and always selects $u_n$ for the second distribution of each pair.

Our third corollary, which gives the lower bounds claimed in lines~1 and~5 of Table~1, is for the ``testing identity, $q$ is known'' problem:

\begin{corollary}
For any $\eps \in (0,1/2]$, there is a constant $c$ such that for sufficiently large $N$ and $1 \leq k=O(\log m)$, there is a $k$-modal distribution $p$ with support $[N]$, and a distribution $D$ over $2(k-1)$-modal distributions of support $[N]$ such that \emph{no} algorithm, when given $c (k \log m)^{1/2}$ samples from a distribution $p' \leftarrow D,$ can distinguish the case that $\dtv(p,p')=0$ from the case $\dtv(p,p')> \eps$ with probability at least $.6$.
\end{corollary}

The above corollary follows from applying Proposition~\ref{prop:lb} to the following trivially verified lower bound construction:

\begin{fact}
Let $D$ be the ensemble of distributions of support $n$ defined as follows:  with probability $1/2,$  $p\leftarrow D$ is the uniform distribution on support $n$, and with probability $1/2$, $p\leftarrow D$ assigns probability $1/2n$ to a random half of the domain elements, and probability $3/2n$ to the other half of the domain elements.  No algorithm, when given fewer than $n^{1/2}/100$ samples from a distribution $p\leftarrow D$ can distinguish between $\dtv(p,u_n)=0$ versus $\dtv(p,u_n)\ge .5$ with probability greater than $.6$.
\end{fact}

As noted previously (after Theorem~\ref{thm:vv11f}), this fact can also be expressed in the language of Proposition~\ref{prop:lb}.

\section{Conclusions}

We have introduced a simple new approach for tackling distribution testing problems for restricted classes
of distributions, by reducing them to general-distribution testing problems over a smaller domain.  We applied this approach to get new testing results  for a range of distribution testing problems involving
monotone and $k$-modal distributions, and established lower bounds showing that all our new
algorithms are essentially optimal.

A general direction for future work is to apply our reduction method to obtain near-optimal testing algorithms for other interesting classes of distributions.  This will involve constructing flat decompositions of
various types of distributions using few samples, which seems to be a natural and interesting
algorithmic problem.  A specific goal is to develop a more efficient version of our
\textsc{ Construct-Flat-Decomposition} algorithm for $k$-modal distributions; is it possible to obtain an
improved version of this algorithm that uses $o(k)$ samples?

\bibliographystyle{alpha} \bibliography{allrefs}

\appendix

\medskip
{\bf
For simplicity, the appendix consists of a slightly expanded and self-contained version of the exposition in the body of the paper, following the ``Notation and Preliminaries'' section.}

\section{Shrinking the domain size:  Reductions for distribution-testing problems} \label{sec:reduction}

In this section we present the general framework of our reduction-based approach and sketch how we
instantiate this approach for monotone and $k$-modal distributions.

\ignore{

{\huge Here's a first stab at this stuff.  I think we should set up a separate section
for this, even if it is short and simple, to highlight that it's a conceptual contribution
of what we do.}

\bigskip

\todo{Note that we now want these flattened and reduced notions where ${\cal I}$ doesn't
cover all of $[n]$ but rather just covers some subset $S$ of $[n]$.  So we'll probably need to change
notation etc a bit below.}

\inote{I changed it so that the decomposition covers all of $[n]$.}
}


We denote by $|I|$ the cardinality of an interval $I \subseteq [n]$, i.e. for $I=[a,b]$ we have
$|I|=b-a+1.$
Fix a distribution $p$ over $[n]$ and a partition of $[n]$ into disjoint intervals
$\mathcal{I} :=  \{ I_i \}_{i=1}^{\ell}.$
The {\em flattened distribution} $(p_f)^{\mathcal{I}}$ corresponding to $p$ and $\mathcal{I}$
is the distribution over $[n]$ defined as follows:  for $j \in [\ell]$ and $i \in I_j$, $(p_f)^{\mathcal{I}} (i) = \littlesum_{t \in I_j} p(t) / |I_j|$.
That is, $(p_f)^{\mathcal{I}}$ is obtained from $p$ by averaging the weight that $p$ assigns to each interval over the entire interval.
The {\em reduced distribution} $(p_r)^{\mathcal{I}}$ corresponding to $p$ and $\mathcal{I}$ is the distribution over $[\ell]$
that assigns the $i$th point the weight $p$ assigns to the
interval $I_i$; i.e., for $i \in [\ell]$, we have $(p_r)^{\mathcal{I}} (i) = p(I_i)$.
Note that if $p$ is non-increasing then so is $(p_f)^{\mathcal{I}}$, but
this is not necessarily the case for $(p_r)^{\mathcal{I}}$.

\medskip\noindent {\bf Definition~\ref{def:flatdecomp}.}\emph{
Let $p$ be a distribution over $[n]$ and let ${\cal I}= \{ I_i \}_{i=1}^{\ell}$ be a partition
of $[n]$ into disjoint intervals.  We say that ${\cal I}$ is a \emph{$(p,\eps,\ell)$-flat decomposition of $[n]$} if
$\dtv(p,(p_f)^{{\cal I}}) \leq \eps.$
}
\medskip

The following useful lemma relates closeness of $p$ and $q$ to closeness of the
reduced distributions:

\medskip
\noindent {\bf Lemma~\ref{lem:reduction}} \emph{
Let ${\cal I} = \{ I_i \}_{i=1}^{\ell}$ be a partition of $[n]$ into disjoint intervals.
Suppose that $p$ and $q$ are distributions over $[n]$ such that ${\cal I}$ is both a $(p,\eps,\ell)$-flat decomposition of $[n]$
and is also a $(q,\eps,\ell)$-flat decomposition of $[n]$.  Then
$ \left| \dtv(p,q) - \dtv ( (p_r)^{\cal I}, (q_r)^{\cal I} ) \right|  \leq 2\eps.$
Moreover, if $p = q$ then $(p_r)^{\cal I} = (q_r)^{\cal I}$.
}

\begin{proof}
The second statement is clear by the definition of a reduced distribution.
To prove the first statement, we first observe that
for any pair of distributions $p, q$ and any partition $\mathcal{I}$ of $[n]$ into disjoint intervals,
we have that $\dtv( (p_r)^{\mathcal{I}} ,(q_r)^{\mathcal{I}}) = \dtv( (p_f)^{\mathcal{I}}, (q_f)^{\mathcal{I}})$.
We thus have that $\left| \dtv(p, q) - \dtv( (p_r)^{\mathcal{I}} , (q_r)^{\mathcal{I}} ) \right|$ is
equal to
\begin{eqnarray*}
\left| \dtv(p, q) - \dtv( (p_f)^{\mathcal{I}} , (q_f)^{\mathcal{I}} ) \right| =
\dtv(p, q) - \dtv( (p_f)^{\mathcal{I}} , (q_f)^{\mathcal{I}} ) \leq \dtv(p,(p_f)^{{\cal I}}) + \dtv(q,(q_f)^{{\cal I}}),
\end{eqnarray*}
where the equality above is equivalent to $\dtv(p,q) \geq \dtv( (p_f)^{\mathcal{I}} , (q_f)^{\mathcal{I}})$
(which is easily verified by considering each interval $I_i \in \mathcal{I}$ separately and applying triangle inequality) and the inequality is the
triangle inequality.
Since $\mathcal{I}$ is both a $(p,\eps,\ell)$-flat decomposition of $[n]$
and a $(q,\eps,\ell)$-flat decomposition of $[n]$, we have that $\dtv(p,(p_f)^{{\cal I}}) \leq \eps$ and $\dtv(q,(q_f)^{{\cal I}}) \leq \eps.$
The RHS above is thus bounded by $2\eps$ and the lemma follows.
\end{proof}

Lemma~\ref{lem:reduction}, while simple, is at the heart of our reduction-based
approach; it lets us transform a distribution-testing problem over the large
domain $[n]$ to a distribution-testing problem over the much smaller ``reduced''
domain $[\ell]$.  At a high level, all our testing algorithms will follow the
same basic approach:  first they run a procedure which, with high probability,
constructs a partition ${\cal I}$ of $[n]$ that is both a $(p,\eps,\ell)$-flat
decomposition of $[n]$ and a $(q,\eps,\ell)$-flat decomposition of
$[n]$.  Next they run the appropriate general-distribution tester over the
$\ell$-element distributions $(p_r)^{{\cal I}},(q_r)^{{\cal I}}$ and output what
it outputs; Lemma~\ref{lem:reduction} guarantees that the distance between
$(p_r)^{\cal I}$ and $(q_r)^{\cal I}$ faithfully reflects the distance between
$p$ and $q$, so this output is correct.

We now provide a few more details that are specific to the various different testing problems that we consider.
For the monotone distribution testing problems the construction of ${\cal I}$ is done obliviously (without drawing any samples or any reference to $p$ or $q$ of any sort) and there is no possibility of failure -- the assumption that $p$ and $q$ are both (say) non-decreasing guarantees that the ${\cal I}$ that is constructed is both a $(p,\eps,\ell)$-flat decomposition of $[n]$ and a $(q,\eps,\ell)$-flat decomposition of $[n]$.
We describe this decomposition procedure in Section~\ref{sec:obliv} and present our monotone distribution testing algorithms that are based on it in Section~\ref{sec:testmonotone}.

For the $k$-modal testing problems it is not so straightforward to construct the desired decomposition ${\cal I}$. This is done via a careful procedure which uses $k^2 \cdot \poly(1/\eps)$ samples from $p$ and $q$.  This procedure has the property that with probability $1-\delta/2$, the ${\cal I}$ it outputs is
both a $(p,\eps,\ell)$-flat decomposition of $[n]$ and a
$(q,\eps,\ell)$-flat decomposition of $[n]$, where $\ell = O(k \log(n)/\eps^2).$\ignore{

\begin{itemize}
\item If $p=q$ then $\Pr[$the procedure outputs ``failure''$] \leq 1/100$; and

\item For any $k$-modal distributions $p,q$ over $[n]$, we have $\Pr[$the procedure outputs an ${\cal I}$ which is not both a $(p,\eps,\ell)$-flat decomposition of $[n]$ and a
$(q,\eps,\ell)$-flat decomposition of $[n]] \leq 1/100.$
\end{itemize}

The first property above ensures that it is ``safe'' to reject the pair $p,q$ if the procedure outputs
``failure,'' and the second property ensures that if the procedure outputs a decomposition ${\cal I}$ then
it is ``safe'' to assume that ${\cal I}$ is both a $(p,\eps,\ell)$-flat decomposition of $[n]$
and a $(q,\eps,\ell)$-flat decomposition of $[n]$.}  Given this, by  running a testing algorithm (which has success probability $1 - \delta/2$) on the
pair $(p_r)^{\cal I},(q_r)^{\cal I}$ of distributions over $[\ell]$, we will get an answer which is with probability $1 - \delta$ a legitimate answer for the original testing problem.  The details are given in Section~\ref{sec:testkmodal}.

We close this section with a result about partitions and
flat decompositions which will be useful later.
Let ${\cal I} = \{I_i\}_{i=1}^{a},$
${\cal I}'=\{I'_j\}_{j=1}^b$ be two partitions of $[n]$. We say
that ${\cal I}'$ is a \emph{refinement} of ${\cal I}$ if for every
$i \in [a]$ there is a subset $S_i$ of $[b]$ such that $\cup_{j \in S_i} I'_j =
I_i$ (note that for this to hold we must have $a \leq b$). Note that $\{S_i\}_{i=1}^a$ forms a partition of $[b]$.
We prove the following useful lemma:

\medskip\noindent{\bf Lemma~\ref{lem:refinement}.}\emph{
Let $p$ be any distribution over $[n]$, let
${\cal I}=\{I_i\}_{i=1}^{a}$ be a $(p,\eps,a)$-flat
decomposition of $[n]$, and let ${\cal J}=\{J_i\}_{i=1}^b$ be
a refinement of ${\cal I}.$ Then ${\cal J}$ is a $(p,2\eps,b)$-flat
decomposition of $[n]$.
}
\medskip

\begin{proof}
Fix any $i \in [\ell]$ and let $S_i \subseteq [b]$ be such that $I_i = \cup_{j \in S_i} J_j.$  To prove the lemma it suffices to show that
\begin{equation}  \label{eq:piece}
2 \littlesum_{t \in I_i}|p(t) - (p_f)^{\cal I}(t)| \geq
\littlesum_{j \in S} \littlesum_{t \in J_j} |p(t) - (p_f)^{{\cal J}}(t)|,
\end{equation}
since the sum on the LHS is the contribution that $I_i$ makes to $\dtv(p,(p_f)^{\cal I})$ and the sum on the RHS is the
corresponding contribution $I_i$ makes to $\dtv(p,(p_f)^{{\cal J}})$.  It may
seem intuitively obvious that the sum on the LHS (which corresponds to approximating the sub-distribution $p^{I_i}$ using a ``global average'') must be smaller
than the sum on the RHS (which corresponds to using separate ``local averages''). However, this intuition is not quite correct, and it is necessary to have the factor of two. To see this, consider a distribution $p$ over $[n]$ such that $p(1) = (1/2) \cdot(1/n)$; $p(i) = 1/n$ for $i \in [2,n-1]$; and $p(n) = (3/2)\cdot (1/n)$. Taking $I_1 = [1,n/2]$ and $I_2=[n/2+1, n]$, it is easy to check that inequality (\ref{eq:piece}) is essentially tight (up to a $o(1)$ factor).

We now proceed to establish (\ref{eq:piece}). Let $T \subseteq [n]$ and consider a partition of $T$ into $k$ nonempty
sets $T_i$, $i \in [k]$. Denote $\mu \eqdef p(T) / |T|$ and  $\mu_i \eqdef p(T_i) / |T_i|$. Then, (\ref{eq:piece}) can be re-expressed as follows
\begin{equation}
\label{eq:re}
2 \littlesum_{t \in T} | p(t) - \mu | \geq \littlesum_{i=1}^k \littlesum_{t \in T_i} |p(t) - \mu_i|.\end{equation}
We shall prove the above statement for all sequences of numbers
$p(1),\dots,p(n)$. Since adding or subtracting the same quantity from each number $p(t)$ does not change the validity of (\ref{eq:re}), for the sake of convenience we may assume all the numbers average to $0$, that is, $\mu = 0$. Consider the $i$-th term on the right hand side,
$\littlesum_{t \in T_i} |p(t) - \mu_i|$. We can bound this quantity from above as follows:
\begin{eqnarray*}
 \littlesum_{t \in T_i} |p(t) - \mu_i| \leq  \littlesum_{t \in T_i} |p(t)|  + |T_i| \cdot |\mu_i|
                                                          =   \littlesum_{t \in T_i} |p(t)| + \left| p(T_i)\right|
                                                         =  2 \littlesum_{t \in T_i} |p(t)|
                                                         =      2 \littlesum_{t \in T_i} |p(t) - \mu|,
\end{eqnarray*}
where the inequality follows from the triangle inequality (applied term by term), the first equality is by the definition of $\mu_i$, the second equality is trivial, and the final equality uses the assumption that $\mu=0$. The lemma follows by summing over $i \in [k]$, using the fact that the $T_i$'s form a partition of $T$.
\end{proof}

\section{Efficiently Testing Monotone Distributions} \label{sec:monotone}

\subsection{Oblivious decomposition of monotone distributions} \label{sec:obliv2}
Our main tool for testing monotone distributions is an {\em oblivious decomposition} of monotone distributions that is a variant of a construction of Birg\'{e}
\cite{Birge:87b}. As we will see it enables us to reduce the problem of
testing a monotone distribution to the problem of testing an arbitrary distribution over a much
smaller domain.  The decomposition result is given below:

\medskip\noindent{\bf Theorem~\ref{thm:birge-oblivious}} (\cite{Birge:87b}){\bf.} \emph{
(oblivious decomposition)  Fix any $n \in \mathbb{Z}^+$ and $\eps>0$.
The partition $\mathcal{I} : = \{ I_i\}_{i=1}^{\ell}$ of $[n]$ described below
 has the following properties:  $\ell = O\left( (1/\eps) \cdot \log (\eps \cdot n + 1) \right)$,
and for any non-increasing distribution $p$ over
$[n]$, ${\cal I}$ is a $(p,O(\eps),\ell)$-flat decomposition of $[n]$.
}
\medskip

There is a dual version of Theorem~\ref{thm:birge-oblivious}, asserting the existence of an ``oblivious''
partition for non-decreasing distributions (which is of course different from the ``oblivious''
partition ${\cal I}$ for non-increasing distributions of Theorem~\ref{thm:birge-oblivious}); this will be useful later.


While our construction is essentially that of Birg\'{e}, we
note that the version given in \cite{Birge:87b} is for
non-increasing distributions over the continuous domain $[0,n]$,
and it is phrased rather differently.  Adapting the arguments
of \cite{Birge:87b} to our discrete setting of
distributions over $[n]$ is not conceptually difficult but requires
some care.  For the sake of being self-contained we
provide a self-contained proof of the discrete version, stated above, that
we require in Appendix~\ref{ap:birge}.

\subsection{Efficiently testing monotone distributions} \label{sec:testmonotone2}

Now we are ready to establish our upper bounds on testing monotone distributions
(given in the first four rows of Table~1).
All of the algorithms are essentially the same: each works by reducing the given monotone distribution testing problem to the same testing problem for arbitrary distributions over support of size $\ell = O(\log n / \eps)$
using the oblivious decomposition from the previous subsection.
For concreteness we explicitly describe the tester for the ``testing identity, $q$ is known'' case below, and then indicate the small changes that are necessary to get the testers for the other three cases.









\medskip

\ignore{\hskip-.2in}
\noindent\framebox{

\medskip \noindent \begin{minipage}{17cm}

\medskip

\textsc{ Test-Identity-Known-Monotone}

\noindent {\bf Inputs:} $\eps,\delta >0$; sample access to non-increasing
distribution $p$ over $[n]$; explicit description of non-increasing distribution
$q$ over $[n]$

\begin{enumerate}
\item Let $\mathcal{I} : = \{I_i\}_{i=1}^{\ell}$, with $\ell= \Theta (\log(\eps n+1)/\eps)$, be the partition of $[n]$ given by Theorem~\ref{thm:birge-oblivious}, which is a $(p',\eps/8,\ell)$-flat decomposition of $[n]$ for any non-increasing distribution $p'$.


\item Let $(q_r)^{\cal I}$ denote the reduced distribution over $[\ell]$ obtained from $q$ using
${\cal I}$ as defined in Section~\ref{sec:reduction}.

\item  Draw $m = s_{IK}(\ell,\eps/2,\delta)$  samples 
        from $(p_r)^{\cal I}$, where $(p_r)^{\cal I}$ is the reduced distribution over $[\ell]$ obtained from $p$ using ${\cal I}$ as defined in Section~\ref{sec:reduction}.


\item Run \textsc{ Test-Identity-Known}$((p_r)^{\cal I}, (q_r)^{\cal I}, {\frac \eps 2},\delta)$ using the samples from Step~3 and
output what it outputs.

\end{enumerate}
\end{minipage}}
\medskip

We now establish our claimed upper bound for the ``testing identity, $q$ is known'' case.
%
%
We first observe that in Step~3, the desired $m=s_{IK}(\ell,\eps/2,\delta)$ samples from $(p_r)^{\cal I}$ can easily be obtained by drawing $m$ samples from $p$ and converting each one to the corresponding draw from
$(p_r)^{\cal I}$ in the obvious way.
If $p = q$ then by Lemma~\ref{lem:reduction} we have that $(p_r)^{\cal I} = (q_r)^{\cal I}$, and \textsc{ Test-Identity-Known-Monotone} outputs ``accept'' with probability at least $1 - \delta$ by Theorem~\ref{thm:testidentityknown}.
If $\dtv(p,q) \geq \eps$, then by Lemma~\ref{lem:reduction} and Theorem~\ref{thm:birge-oblivious} we have that $\dtv((p_r)^{\cal I},(q_r)^{\cal I}) \geq 3\eps/4$, so \textsc{ Test-Identity-Known-Monotone} outputs ``reject'' with probability at least $1 - \delta$ by Theorem~\ref{thm:testidentityknown}.
For the ``testing identity, $q$ is unknown'' case, the the algorithm \textsc{ Test-Identity-Unknown-Monotone} is very similar to \textsc{ Test-Identity-Known-Monotone}.  The differences are as follows:  instead of Step~2, in Step~3 we draw
$m = s_{IU}(\ell,\eps/2,\delta)$ samples from $(p_r)^{\cal I}$ and the same number
of samples from $(q_r)^{\cal I}$; and in Step~4, we run \textsc{ Test-Identity-Unknown}$((p_r)^{\cal I},(q_r)^{\cal I},
{\frac \eps 2},\delta)$ using the samples from Step~3.
The analysis is exactly the same as above (using Theorem~\ref{thm:testidentityunknown} in place
of Theorem~\ref{thm:testidentityknown}).

We now describe the algorithm \textsc{$L_1$-Estimate-Known-Monotone} for the ``tolerant testing, $q$ is known'' case. This algorithm takes values $\eps$ and $\delta$ as input, so the partition ${\cal I}$ defined in Step~1 is a $(p',\eps/4,\ell)$-flat decomposition of $[n]$ for any non-increasing $p'.$
In Step~3 the algorithm draws $m = s_{E}(\ell,\eps/2,\delta)$ samples
and runs \textsc{ $L_1$-Estimate}$((p_r)^{\cal I}, (q_r)^{\cal I},\eps/2,\delta)$ in Step~4.
If $\dtv(p,q)=c$ then by Lemma~\ref{lem:reduction} we have that $\dtv((p_r)^{\cal I},(q_r)^{\cal I}) \in [c-\eps/2,c+\eps/2]$ and \textsc{ $L_1$-Estimate-Known-Monotone} outputs a value within the prescribed range with probability at least $1-\delta,$ by Theorem~\ref{thm:testtolerantknown}.
The algorithm \textsc{ $L_1$-Estimate-Unknown-Monotone}
case and its analysis are entirely similar.

\section{Efficiently Testing $k$-modal Distributions} \label{sec:testkmodal}

In this section we establish our main positive testing results for $k$-modal distributions, the upper bounds
stated in the final four rows of Table~1.
The key to all these results is an algorithm \textsc{ Construct-Flat-Decomposition}$(p,\eps,\delta)$.
We prove the following performance guarantee about this algorithm:

\medskip\noindent{\bf Lemma~\ref{lem:cfsd}.}\emph{
Let $p$ be a $k$-modal distribution over $[n]$.  Algorithm \textsc{ Construct-Flat-Decomposition}$(p,\eps,\delta)$
draws $O(k^2 \eps^{-4} \log(1/\delta))$ samples from $p$ and outputs a $(p,\eps,\ell)$-flat decomposition of $[n]$ with probability at least $1-\delta$, where $\ell = O(k \log(n)/\eps^2)$.
}
\medskip

The bulk of our work in Section~\ref{sec:testkmodal} is to describe \textsc{ Construct-Flat-Decomposition}$(p,\eps,\delta)$ and prove Lemma~\ref{lem:cfsd}, but first we show how Lemma~\ref{lem:cfsd} easily yields our claimed testing results for $k$-modal distributions.
As in the monotone case all four algorithms are essentially the same:  each works by reducing the given $k$-modal distribution testing problem to the same testing problem for arbitrary distributions over $[\ell].$  We describe the \textsc{ Test-Identity-Known-kmodal} algorithm below, and then indicate the necessary changes to get the other three testers.

The following terminology will be useful:  Let ${\cal I}= \{I_i\}_{i=1}^r$ and ${\cal I}'=\{I'_i\}_{i=1}^s$ be two
partitions of $[n]$ into $r$ and $s$ intervals respectively.  The \emph{common refinement} of ${\cal I}$
 and ${\cal I}'$ is the partition ${\cal J}$ of $[n]$ into intervals obtained from ${\cal I}$ and ${\cal I}'$ in the obvious way, by taking all possible nonempty intervals of the form
 $I_i \cap I'_j.$  It is clear that ${\cal J}$ is both a refinement of ${\cal I}$ and of ${\cal I}'$
 and that the number of intervals $|{\cal J}|$ in ${\cal J}$ is at most $r+s.$


\ignore{(Note that it uses samples from $q$
even though it has an explicit description of $q$; this is for consistency with the
later \textsc{ Test-kmodal-Unknown} algorithm, and does not hurt the sample
complexity of \textsc{ Test-kmodal-Known}.)}

\medskip

\ignore{\hskip-.2in}

\noindent \framebox{

\medskip \noindent \begin{minipage}{17cm}

\medskip

\textsc{ Test-Identity-Known-kmodal}

\noindent {\bf Inputs:} $\eps,\delta >0$; sample access to $k$-modal
distributions $p,q$ over $[n]$

\begin{enumerate}

\item Run \textsc{ Construct-Flat-Decomposition}$(p,\eps/2,\delta/4)$ and let ${\cal I}=\{I_i\}_{i=1}^{\ell},$
    $\ell = O(k\log(n)/\eps^2)$, be the partition that it outputs.  Run
    \textsc{ Construct-Flat-Decomposition}$(p,\eps/2,\delta/4)$ and let ${\cal I}'=\{I'_i\}_{i=1}^{\ell'},$
    $\ell' = O(k\log(n)/\eps^2)$, be the partition that it outputs.  Let ${\cal J}$ be the
    common refinement of ${\cal I}$ and ${\cal I}$ and
    let $\ell_{\cal J} = O(k \log(n)/\eps^2)$ be the number of intervals in ${\cal J}.$

\item Let $(q_r)^{\cal J}$ denote the reduced distribution over $[\ell_{\cal J}]$ obtained from $q$
using ${\cal J}$ as defined in Section~\ref{sec:reduction}.

\item  Draw $m = s_{IK}(\ell_{\cal J},\eps/2,\delta/2)$ samples 
        from $(p_r)^{\cal J}$, where $(p_r)^{\cal J}$ is the reduced distribution over $[\ell_{\cal J}]$ obtained from $p$ using ${\cal J}$ as defined in Section~\ref{sec:reduction}.


\item Run \textsc{ Test-Identity-Known}$((p_r)^{\cal J}, (q_r)^{\cal J}, {\frac \eps 2},{\frac \delta 2})$ using the samples from Step~3 and output what it outputs.


\end{enumerate}
\end{minipage}}
\medskip

We note that Steps~2, 3 and~4 of \textsc{ Test-Identity-Known-kmodal} are the same as the corresponding
steps of \textsc{ Test-Identity-Known-Monotone}.  For the analysis of \textsc{ Test-Identity-Known-kmodal},
 Lemmas~\ref{lem:cfsd} and~\ref{lem:refinement} give us that with probability $1 - \delta/2$, the partition
${\cal J}$ obtained in Step~1 is both a $(p,\eps,\ell_{\cal J})$-flat and $(q,\eps,\ell_{\cal J})$-flat decomposition of $[n]$;
we condition on this going forward.  From this point on the analysis is essentially identical to
the analysis for \textsc{ Test-Identity-Known-Monotone}
and is omitted.

The modifications required to obtain
algorithms \textsc{ Test-Identity-Unknown-kmodal}, \textsc{ $L_1$-Estimate-Known-kmodal} and
\textsc{ $L_1$-Estimate-Unknown-kmodal}, and the analysis of these algorithms, are completely analogous to the
modifications and analyses of Appendix~\ref{sec:testmonotone2} and are omitted.

\subsection{The \textsc{ Construct-Flat-Decomposition} algorithm.}

We present \textsc{ Construct-Flat-Decomposition}$(p,\eps,\delta)$ followed by an intuitive explanation.  Note that it employs a procedure \textsc{ Orientation}$(\wh{p},I)$, which uses no samples and is presented and analyzed in Section~\ref{sec:kmodaltools}.

\medskip

\noindent\framebox{

\medskip \noindent \begin{minipage}{17cm}

\medskip

\textsc{ Construct-Flat-Decomposition}

\noindent \textsc{ Inputs:} $\eps,\delta >0$; sample access to $k$-modal
distribution $p$ over $[n]$

\begin{enumerate}

\vspace{-.2cm} \item Initialize ${\cal I} := \emptyset.$

\vspace{-.2cm} \item Fix $\tau:=\eps^2/(20000k)$. Draw $ r = \Theta(\log(1/\delta)/\tau^2)$ samples
from $p$ and let $\wh{p}$ denote the resulting empirical distribution (which by Theorem~\ref{thm:DKW} has
$\dk(\wh{p},p) \leq \tau$ with probability at least $1-\delta$).

\vspace{-.2cm} \item Greedily partition the domain $[n]$ into $\alpha$ {\em atomic intervals} $\{I_i\}_{i=1}^{\alpha}$ as follows:
$I_1 := [1, j_1]$, where $j_1 := \min\{ j \in [n] \mid \wh{p}([1,j]) \geq \eps/(100k)\}$.
For $i \geq 1$, if $\cup_{j=1}^i I_j = [1, j_i]$, then
$I_{i+1}:=[j_i+1, j_{i+1}]$, where $j_{i+1}$ is defined as follows:
If $\wh{p}([j_i+1, n]) \geq \eps/(100k)$, then $j_{i+1}: = \min\{ j \in [n]
\mid \wh{p}([j_{i}+1, j]) \geq \eps/(100k)\}$, otherwise, $j_{i+1} := n$.

\vspace{-.2cm} \item Construct a set of $n_m$ {\em moderate intervals}, a set of
$n_h$ {\em heavy points}, and a set of $n_n$ {\em
negligible intervals} as follows:  For each atomic interval $I_i = [a,b] $,
\begin{itemize}

\vspace{-.2cm} \item [(a)] if $\wh{p}([a,b]) \leq 3\eps/(100k)$ then $I_i$ is declared to be
a \emph{moderate} interval;

\vspace{-.2cm} \item [(b)] otherwise we have $\wh{p}([a,b]) > 3\eps/(100k)$ and we declare
$b$ to be a \emph{heavy point}.  If $a < b$ then we declare $[a,b-1]$
to be a \emph{negligible} interval.
\end{itemize}

For each interval $I$ which is a heavy point, add $I$ to ${\cal I}.$
Add each negligible interval $I$  to ${\cal I}.$

\vspace{-.2cm} \item For each moderate interval $I$, run procedure \textsc{ Orientation}$(\wh{p},I)$; let $\circ \in \{\uparrow,\downarrow,\bot\}$ be its output.

     If $\circ = \bot$ then add $I$  to ${\cal I}.$

    If $\circ = \downarrow$ then let ${\cal J}_I$ be the partition of $I$ given
    by Theorem~\ref{thm:birge-oblivious} which is a $(p',\eps/4,O(\log(n)/\eps))$-flat decomposition of
    $I$ for any non-increasing distribution $p'$ over $I.$  Add all the elements of ${\cal J}_I$ to ${\cal I}.$

    If $\circ = \uparrow$ then let ${\cal J}_I$ be the partition of $I$ given
    by the dual version of Theorem~\ref{thm:birge-oblivious}, which is a $(p',\eps/4,O(\log(n)/\eps))$-flat decomposition of $I$ for any non-decreasing distribution $p'$ over $I.$ Add all the elements of ${\cal J}_I$ to ${\cal I}.$

\vspace{-.2cm} \item Output the partition ${\cal I}$ of $[n]$.

\end{enumerate}
\end{minipage}}

\medskip

Roughly speaking, when \textsc{ Construct-Flat-Decomposition} constructs a partition ${\cal I}$, it initially breaks $[n]$ up into
two types of intervals.  The first type are intervals that are ``okay'' to include in
a flat decomposition, either because they have very little mass, or because they
consist of a single point, or because they are close to uniform.  The second type are intervals
that are ``not okay'' to include in a flat decomposition -- they have significant mass
and are far from uniform -- but the algorithm is able to ensure that almost all of these are monotone
distributions with a known orientation.  It then uses the oblivious decomposition of
Theorem~\ref{thm:birge-oblivious} to construct a flat decomposition of each such interval.  (Note that
it is crucial that the orientation is known in order to be able to use Theorem~\ref{thm:birge-oblivious}.)

In more detail, \textsc{ Construct-Flat-Decomposition}$(p,\eps,\delta)$
works  as follows.  The algorithm first draws a batch of samples from $p$ and uses them to construct an
estimate $\wh{p}$ of the CDF of $p$ (this is straightforward using the DKW
inequality). Using $\wh{p}$ the algorithm partitions $[n]$ into a collection
of $O(k/\eps)$ disjoint intervals in the following way:
\begin{itemize}
\vspace{-.2cm} \item A small collection of the intervals are ``negligible''; they collectively have total mass less than $\eps$ under $p$.  Each negligible interval $I$ will be an element of the partition ${\cal I}.$

\vspace{-.2cm} \item Some of the intervals are ``heavy points''; these are intervals consisting of a single point
that has mass $\Omega(\eps/k)$ under $p$.  Each heavy point $I$ will also be an element of the
partition ${\cal I}.$

\vspace{-.2cm} \item The remaining intervals are ``moderate'' intervals, each of which has mass $\Theta(\eps/k)$ under $p$.
\end{itemize}

It remains to incorporate the moderate intervals into the partition ${\cal I}$ that is being constructed.
This is done as follows:  using $\wh{p}$, the algorithm comes up with a ``guess'' of the correct orientation (non-increasing, non-decreasing, or close to uniform) for each moderate interval.  Each moderate interval where the ``guessed'' orientation is ``close to uniform'' is included in the partition ${\cal I}.$ Finally, for each moderate interval $I$ where the guessed orientation is ``non-increasing'' or ``non-decreasing'', the algorithm invokes Theorem~\ref{thm:birge-oblivious} on $I$ to perform the oblivious decomposition for monotone distributions, and the resulting sub-intervals are included in ${\cal I}$.
The analysis will show that the guesses are almost always
correct, and intuitively this should imply that  the ${\cal I}$ that is constructed is indeed a $(p,\eps,\ell)$-flat decomposition of $[n]$.

\subsection{Performance of \textsc{ Construct-Flat-Decomposition}:  Proof of Lemma~\ref{lem:cfsd}.}





%


The claimed sample bound is obvious from
inspection of the algorithm, as the only step that draws any samples is Step~2. The bound on the number of intervals in the flat decomposition follows directly from the upper bounds on the number of heavy points, negligible intervals and moderate intervals shown below, using also Theorem~\ref{thm:birge-oblivious}.
It remains to show that the output of the algorithm is a valid flat decomposition of $p$. First, by the DKW inequality (Theorem~\ref{thm:DKW}) we have that with probability at least $1 - \delta$
it is the case that
\begin{align}
|\wh{p}(I)-p(I)| \leq {\frac {\eps^2}{10000k}},~~\text{for every interval~}I \subseteq [n]. \label{eq:result of DKW}
\end{align}

We make some preliminary observations about the weight that $p$ has on the intervals
constructed in Steps~4 and~5.
Since every atomic interval $I_i$ constructed in
Step~4 has $\wh{p}(I) \geq \eps/(100k)$ (except potentially the rightmost one),
it follows that the number $\alpha$ of atomic intervals constructed in Step~3
satisfies $$\alpha \leq \lceil  100 k/\eps \rceil.$$

We now establish bounds on the probability mass that $p$ assigns
to the moderate intervals, heavy points, and negligible intervals
that are constructed in Step~4.  Using~\eqref{eq:result of DKW},
each interval $I_i$ that is declared to be a moderate interval
in Step~4(a) must satisfy
\begin{equation}
99\eps/(10000k) \leq p([a,b]) \leq 301 \eps/(10000k) \quad \text{(for all
moderate intervals $[a,b]$)}. \label{eq:moderate}
\end{equation}

By virtue of the greedy process that is used to construct atomic intervals
in Step~3, each point $b$ that is declared to be a heavy point in
Step~4(b) must satisfy $\wh{p}(b) \geq 2\eps/(100k)$ and thus using~\eqref{eq:result of DKW} again
\begin{equation}
p(b) \geq 199 \eps/(10000k) \quad \text{(for all heavy points $b$)}.
\label{eq:heavy}
\end{equation}
Moreover, each interval $[a,b-1]$ that is
declared to be a negligible interval must satisfy $\wh{p}([a,b-1])
< \eps/(100k)$ and thus using~\eqref{eq:result of DKW} again
\begin{equation}
p([a,b-1]) \leq 101 \eps / (10000k)
\quad \text{(for all negligible intervals $[a,b-1]$)}.
\label{eq:negligible}
\end{equation}

It is clear that $n_m$ (the number of moderate intervals) and $n_h$ (the number of heavy
points) are each at most $\alpha.$
Next we observe that the number of negligible intervals $n_n$ satisfies
\[n_n \leq k.\]
This is because
at the end of each negligible interval $[a,b-1]$ we have (observing
that each negligible interval must be nonempty) that
$p(b-1) \leq p([a,b-1]) \leq 101\eps/(10000k)$ while $p(b) \geq 199\eps/(10000k)$.
Since $p$ is $k$-modal, there can be at most $\lceil (k+1)/2
\rceil \leq k$ points $b \in [n]$ satisfying this condition. Since each
negligible interval $I$ satisfies $p(I) \leq 101\eps / (10000k)$ we have
that the total probability mass under $p$ of all the negligible intervals
is at most $101 \eps / 10000.$

Thus far we have built a partition of $[n]$  into a collection of $n_m \leq \lceil 100k/\eps \rceil$
moderate intervals (which we denote $M_1,\dots,M_{n_m}$), a set of $n_h \leq \lceil 100k/\eps \rceil$
heavy points (which we denote $h_1,\dots,h_{n_h}$) and a set of $n_n \leq k$ negligible
intervals (which we denote $N_1,\dots,N_{n_n}$).
Let $A \subseteq \{1,\dots,n_m\}$ denote the set of those indices $i$ such that \textsc{ Orientation}$(\wh{p},M_i)$ outputs $\bot$ in Step~6.
The partition ${\cal I}$ that \textsc{ Construct-Flat-Decomposition} constructs consists of
$\{h_1\},\dots,\{h_{n_h}\}$, $N_1,\dots,N_{n_n}$, $\{M_i\}_{i \in A}$, and
$
\bigcup_{i \in ([n_m] \setminus A)} {\cal J}_{M_i}.
$
We can thus write $p$ as
\begin{equation} \label{eq:p}
p = \littlesum_{j=1}^{n_h} p(h_j) \cdot {\bm 1}_{h_j} + \littlesum_{j=1}^{n_n} p(N_j) p_{N_j} +
\littlesum_{j \in A} p(M_j) p_{M_j} +
\littlesum_{j \in ([n_m] \setminus A)} \littlesum_{I \in {\cal J}_{M_j}} p(I)p_I.
\end{equation}
Using Lemma~\ref{lem:decomp} (proved in Appendix~\ref{ap:boundvar}) we can bound the total variation distance between
$p$ and $(p_f)^{\cal I}$ by
\begin{eqnarray}
\dtv(p,(p_f)^{\cal I}) &\leq&  {\frac 1 2}\littlesum_{j=1}^{n_h} |  p(h_j) - (p_f)^{\cal I}(h_j)| +
{\frac 1 2} \littlesum_{j=1}^{n_n} |p (N_j)-(p_f)^{\cal I}(N_j)|  +
\littlesum_{j=1}^{n_n} p(N_j) \cdot \dtv(p_{N_j},((p_f)^{\cal I})_{N_j}) \nonumber \\
&& +
{\frac 1 2} \littlesum_{j \in A} |p(M_j) - (p_f)^{\cal I}(M_j)| +
\littlesum_{j \in A} p(M_j) \cdot \dtv(p_{M_j},((p_f)^{\cal I})_{M_j}) \nonumber \\
&& +
{\frac 1 2} \littlesum_{j \in ([n_m] \setminus A)} \littlesum_{I \in J_{M_j}}
|p(I) - (p_f)^{\cal I}(I)| +
\littlesum_{j \in ([n_m] \setminus A)} \littlesum_{I \in J_{M_j}}
p(I) \cdot \dtv(p_I,((p_f)^{\cal I})_I).  \label{eq:p2}
\end{eqnarray}
Since $p(I) = (p_f)^{\cal I}(I)$ for every $I \in {\cal I}$, this simplifies to
\begin{eqnarray}
\dtv(p,(p_f)^{\cal I}) &\leq&
\littlesum_{j=1}^{n_n} p(N_j) \cdot \dtv(p_{N_j},((p_f)^{\cal I})_{N_j}) +
\littlesum_{j \in A} p(M_j) \cdot \dtv(p_{M_j},((p_f)^{\cal I})_{M_j}) \nonumber \\
&& +
\littlesum_{j \in ([n_m] \setminus A)} \littlesum_{I \in J_{M_j}}
p(I) \cdot \dtv(p_I,((p_f)^{\cal I})_I).  \label{eq:p3}
\end{eqnarray}
which we now proceed to bound.

Recalling from (\ref{eq:negligible}) that $p(N_j) \leq 101\eps/(10000k)$ for each
negligible interval $N_j$, and recalling that $n_n \leq k$, the first summand in (\ref{eq:p3}) is at most $101\eps/10000.$

To bound the second summand, fix any $j \in A$ so $M_j$ is a moderate interval such that
\textsc{ Orientation}$(\wh{p},M_j)$ returns $\bot.$ If $p_{M_j}$ is non-decreasing then by Lemma~\ref{claim:orientation} it must be the case that $\dtv(p_{M_j},((p_f)^{\cal I})_{M_j}) \leq \eps/6$
(note that $((p_f)^{\cal I})_{M_j}$ is just $u_{M_j}$, the uniform distribution over $M_j$).
Lemma~\ref{claim:orientation} gives the same bound if $p_{M_j}$ is non-increasing.  If $p_{M_j}$ is neither non-increasing nor non-decreasing then we have no nontrivial bound on $\dtv(p_{M_j},((p_f)^{\cal I})_{M_j})$, but since $p$ is $k$-modal there can be at most $k$ such values of $j$ in $A$.  Recalling (\ref{eq:moderate}), overall we have that
\[
\littlesum_{j \in A} p(M_j) \cdot \dtv(p_{M_j},((p_f)^{\cal I})_{M_j}) \leq {\frac {301 \eps k}{10000k}} +
{\frac \eps 6} \leq {\frac {1968 \eps}{10000}},
\]
and we have bounded the second summand.

It remains to bound the final summand of (\ref{eq:p3}).
For each $j \in ([n_m] \setminus A)$, we know that \textsc{ Orientation}$(\wh{p},M_j)$
outputs either $\uparrow$ or $\downarrow$.\ignore{\cnote{The previous argument given here was, I think, insufficient. Here it is in case someone wants to read it/I'm mistaken: By Lemma~\ref{claim:orientation} we have that the output of
\textsc{ Orientation}$(\wh{p},M_j)$ gives the correct orientation of $p_{M_j}$.  Consequently ${\cal J}_{M_j}$
is a $(p_{M_j},\eps/4,O(\log(n)/\eps))$-flat decomposition of $M_j$, by Theorem~\ref{thm:birge-oblivious}. This means that
$\dtv(p_{M_j},((p_f)^{\cal I})_{M_j}) \leq \eps/4$, which is equivalent to
\[
{\frac 1 {p(M_j)}} \littlesum_{I \in {\cal J}_{M_j}} p(I) \dtv(p_I,((p_f)^{\cal I})_I) \leq {\frac \eps 4}, \quad \text{i.e.}\quad
\littlesum_{I \in {\cal J}_{M_j}} p(I) \dtv(p_I,((p_f)^{\cal I})_I) \leq p(M_j) \cdot {\frac \eps 4}.
\]

Summing over all $j \in ([n_m] \setminus A)$, we get
\[
\littlesum_{j \in ([n_m] \setminus A)}  \littlesum_{I \in {\cal J}_{M_j}} p(I) \dtv(p_I,((p_f)^{\cal I})_I) \leq \littlesum_{j \in ([n_m] \setminus A)} p(M_j) \cdot {\frac \eps 4} \leq {\frac \eps 4}.
\]
So the third summand of (\ref{eq:p3}) is at most $\eps/4$, and overall we have that
$(\ref{eq:p3}) \leq {\frac {4468 \eps}{10000}}$. Hence, we have shown that ${\cal I}$ is a
$(p,\eps,\ell)$-flat decomposition of $[n]$.}} If $p_{M_j}$ is monotone, then by Lemma~\ref{claim:orientation} we have that the output of
\textsc{ Orientation}$(\wh{p},M_j)$ gives the correct orientation of $p_{M_j}$.  Consequently ${\cal J}_{M_j}$
is a $(p_{M_j},\eps/4,O(\log(n)/\eps))$-flat decomposition of $M_j$, by Theorem~\ref{thm:birge-oblivious}. This means that
$\dtv(p_{M_j},((p_f)^{\cal I})_{M_j}) \leq \eps/4$, which is equivalent to
\[
{\frac 1 {p(M_j)}} \littlesum_{I \in {\cal J}_{M_j}} p(I) \dtv(p_I,((p_f)^{\cal I})_I) \leq {\frac \eps 4}, \quad \text{i.e.}\quad
\littlesum_{I \in {\cal J}_{M_j}} p(I) \dtv(p_I,((p_f)^{\cal I})_I) \leq p(M_j) \cdot {\frac \eps 4}.
\]
Let $B \subset [n_m] \setminus A$ be such that, for all $j \in B$, $p_{M_j}$ is monotone. Summing the above over all $j \in B$ gives:
\[
\littlesum_{j \in B}  \littlesum_{I \in {\cal J}_{M_j}} p(I) \dtv(p_I,((p_f)^{\cal I})_I) \leq \littlesum_{j \in B} p(M_j) \cdot {\frac \eps 4} \leq {\frac \eps 4}.
\]
Given that $p$ is $k$-modal, the cardinality of the set $[n_m]\setminus (A \cup B)$ is at most $k$. So we have the bound:
\[
\littlesum_{j \in [n_m]\setminus (A \cup B)}  \littlesum_{I \in {\cal J}_{M_j}} p(I) \dtv(p_I,((p_f)^{\cal I})_I) \leq \littlesum_{j \in [n_m]\setminus (A \cup B)} p(M_j)  \leq {\frac {301 \eps k}{10000k}}.
\]
So the third summand of (\ref{eq:p3}) is at most $\eps/4+301 \eps/10000$, and overall we have that
$(\ref{eq:p3}) \leq {\frac { \eps}{2}}$. Hence, we have shown that ${\cal I}$ is a
$(p,\eps,\ell)$-flat decomposition of $[n]$, and Lemma~\ref{lem:cfsd} is proved.

\subsection{The \textsc{ Orientation} algorithm.} \label{sec:kmodaltools2}

The \textsc{ Orientation} algorithm takes as input an explicit distribution of a distribution $\wh{p}$ over
$[n]$ and an interval $I \subseteq [n].$  Intuitively, it assumes that $\wh{p}_I$ is close
(in Kolmogorov distance) to a monotone distribution $p_I$, and its goal is to determine the orientation of $p_I$: it outputs either $\uparrow$, $\downarrow$ or $\bot$ (the last of which means ``close to uniform'').
The algorithm is quite simple; it checks whether there exists an
initial interval $I'$ of $I$ on which $\wh{p}_I$'s weight is significantly different from $u_I(I')$ (the weight
that the uniform distribution over $I$ assigns to $I'$) and bases its output on this in the obvious way.  A precise description of the algorithm (which uses no samples) is given below.

\bigskip

\ignore{\hskip-.2in} \framebox{

\medskip \noindent \begin{minipage}{16cm}

\medskip

\textsc{ Orientation}

\noindent \textsc{ Inputs:} explicit description of distribution $\wh{p}$ over $[n]$; interval
$I =[a,b] \subseteq [n]$

\smallskip

\begin{enumerate}

\item If $|I|=1$ (i.e. $I = \{a\}$ for some $a \in [n]$) then return ``$\bot$'', otherwise continue.

\item If there is an initial interval $I'=[a,j]$ of $I$ that satisfies
$
u_I(I')-(\wh{p})_I(I') > {\frac \eps 7}
$
then halt and output ``$\uparrow$''.  Otherwise,

\item If there is an initial interval $I'=[a,j]$ of $I$ that satisfies
$
u_I(I')-(\wh{p})_I(I') < - {\frac \eps 7}$
then halt and output ``$\downarrow$''.  Otherwise,

\item Output ``$\bot$''.

\end{enumerate}

\smallskip

\end{minipage}}
\bigskip

We proceed to analyze \textsc{ Orientation}. We show that if $p_I$ is far from uniform then \textsc{ Orientation} outputs the correct orientation for it. We also show that whenever \textsc{ Orientation} does not output ``$\bot$'', whatever it outputs is the correct orientation of $p_I$. For ease of readability, for the rest of this subsection we use the following notation:
\[
\Delta := {\frac {\eps^2}{10000k}}
\]

\medskip\noindent{\bf Lemma~\ref{claim:orientation}.}\emph{
Let $p$ be a distribution over $[n]$ and let interval $I =[a,b] \subseteq [n]$ be such that $p_I$ is monotone.
Suppose $p(I) \geq 99\eps/(10000k)$, and suppose that for every interval $I' \subseteq I$ we have that}
\begin{equation}
|\wh{p}(I')-p(I')| \leq \Delta.
\label{eq:nifty}
\end{equation}
\emph{Then}
\begin{enumerate}
\item \emph{If $p_I$ is non-decreasing and $p_I$ is $\eps/6$-far from the uniform distribution
$u_I$ over $I$, then \textsc{ Orientation}$(\wh{p},I)$ outputs ``$\uparrow$'';}

\item \emph{if \textsc{ Orientation}$(\wh{p},I)$ outputs ``$\uparrow$'' then $p_I$ is non-decreasing;}

\item \emph{if $p_I$ is non-increasing and $p_I$ is $\eps/6$-far from the uniform distribution
$u_I$ over $I$, then \textsc{ Orientation}$(\wh{p},I)$ outputs ``$\downarrow$'';}

\item \emph{if \textsc{ Orientation}$(\wh{p},I)$ outputs ``$\downarrow$'' then $p_I$ is non-increasing.}
\end{enumerate}

\medskip

\begin{proof}
Let $I'=[a,j] \subseteq I$ be any initial interval of $I.$  We first establish the upper bound
\begin{equation}
|p_I(I')-(\wh{p})_I(I')| \leq \eps/49
\label{eq:closeyclose}
\end{equation}
as this will be useful for the rest of the proof.  Using (\ref{eq:nifty}) we have
\begin{eqnarray}
p_I(I')-(\wh{p})_I(I') &=& {\frac {p(I')}{p(I)}} - {\frac {\wh{p}(I')}{\wh{p}(I)}} \geq
{\frac {p(I')}{p(I)}} - {\frac {p(I') + \Delta}{{p}(I) - \Delta}} \nonumber \\
&=& - \Delta \cdot {\frac {p(I')+p(I)}{p(I)(p(I)- \Delta)}}.
\label{eq:bib}
\end{eqnarray}
Now using the fact that $p(I') \leq p(I)$ and $p(I) \geq 99 \eps/(10000k)$, we get that (\ref{eq:bib})
is at least
\[
- \Delta \cdot {\frac {2p(I)}{(98/99)p(I)^2}} =
-{\frac {2 \cdot 99 \Delta}{98 p(I)}} \geq
-{\frac {2 \cdot 99 \Delta \cdot 10000k}{98 \cdot  99 \eps}}
= -{\frac {\eps}{49}}.
\]
So we have established the lower bound $p_I(I')-(\wh{p})_I(I') \geq -\eps/49.$
For the upper bound, similar reasoning gives
\begin{eqnarray*}
p_I(I')-(\wh{p})_I(I') &\leq& \Delta \cdot {\frac {p(I')+p(I)}{p(I)(p(I) + \Delta)}}
\leq \Delta \cdot {\frac {2p(I)}{p(I)^2 \cdot (100/99)}}\\
&\leq& \Delta \cdot {\frac {2 \cdot 10000k \cdot 99}{99 \eps \cdot 100}} = {\frac \eps {50}}
\end{eqnarray*}
and so we have shown that $|p_I(I')-(\wh{p})_I(I')| \leq \eps/49$ as desired.  Now we proceed to prove the lemma.

We first prove Part 1. Suppose that $p_I$ is non-decreasing and $\dtv(p_I,u_I) > \eps/6.$
Since $p_I$ is monotone and $u_I$ is uniform and both are supported on $I$, we have that the
pdfs for $p_I$ and $u_I$ have exactly one crossing. An easy consequence of this is that
$\dk(p_I,u_I) = \dtv(p_I,u_I) > \eps/6.$  By the definition of $d_K$ and the fact that $p_I$ is
non-decreasing, we get that there exists a point $j \in I$ and an interval $I'=[a,j]$
which is such that
\[
\dk(p_I,u_I) = u_I(I') - p_I(I') > {\frac \eps 6}.
\]
Using (\ref{eq:closeyclose}) we get from this that
\[
u_I(I')-(\wh{p})_I(I') > {\frac \eps 6} - {\frac \eps {49}} > {\frac \eps 7}
\]
and thus \textsc{ Orientation} outputs ``$\uparrow$'' in Step~3 as claimed.

Now we turn to Part 2 of the lemma. Suppose that \textsc{ Orientation}$(\wh{p},I)$ outputs ``$\uparrow$''. Then it must be the case that there is an initial interval $I'=[a,j]$ of $I$ that satisfies $u_I(I')-(\wh{p})_I(I') > {\frac \eps 7}.$  By (\ref{eq:closeyclose}) we have that $u_I(I') - p_I(I') > {\frac \eps 7} - {\frac \eps {49}} = {\frac {6 \eps}{49}}$. But Observation~\ref{obs:mon} tells us that if $p_I$ were non-increasing then we would have $u_I(I')-p_I(I') \leq 0$; so $p_I$ cannot be non-increasing, and therefore it must be non-decreasing.

For Part 3, suppose that $p_I$ is non-increasing and $\dtv(p_I,u_I) > \eps/6.$
First we must show that \textsc{ Orientation} does \emph{not} output ``$\uparrow$'' in Step~3.
Since $p_I$ is non-increasing, Observation~\ref{obs:mon} gives us that
$u_I(I') - p_I(I') \leq 0$ for every initial interval $I'$ of $I$.  Inequality (\ref{eq:closeyclose})
then gives $u_I(I') - (\wh{p})_I(I') \leq \eps/49$, so \textsc{ Orientation} indeed does not output ``$\uparrow$''
in Step~3 (and it reaches Step~4 in its execution).  Now arguments exactly analogous to the arguments
for part 1 (but using now the fact that $p_I$ is non-increasing rather than non-decreasing) give that
there is an initial interval $I'$ such that
$(\wh{p})_I(I') - u_I(I') > {\frac \eps 6} - {\frac \eps {49}} > {\frac \eps 7}$, so \textsc{ Orientation}
outputs ``$\downarrow$'' in Step~4 and Part 3 of the lemma follows.

Finally, Part 4 of the lemma follows from analogous arguments as Part 2.\end{proof}

\section{Proof of Proposition~\ref{prop:lb}} \label{sec:proplb}

We start by defining the transformation, and then prove the necessary lemmas to show that the transformation yields $k$-modal distributions with the specified increase in support size, preserves $L_1$ distance between pairs, and has the property that samples from the transformed distributions can be simulated given access to samples from the original distributions.

The transformation proceeds in two phases.  In the first phase, the input distribution $p$ is transformed into a related distribution $f$ with larger support; $f$ has the additional property that the ratio of the probabilities of consecutive domain elements is bounded.  Intuitively the distribution $f$ corresponds to a ``reduced distribution'' from Section~\ref{sec:reduction}.  In the second phase, the distribution $f$ is transformed into the final $2(k-1)$-modal distribution $g$.  Both stages of the transformation consist of subdividing each element of the domain of the input distribution into a set of elements of the output distribution;  in the first stage, the probabilities of each element of the set are chosen according to a geometric sequence, while in the second phase, all elements of each set are given equal probabilities.

We now define this two-phase transformation and prove Proposition~\ref{prop:lb}.

\begin{definition}~\label{def:lb}
Fix $\eps>0$ and a distribution $p$ over $[n]$ such that $p_{min} \leq p(i) \leq p_{max}$ for all $i \in [n].$  We define the distribution $f_{p,\eps,p_{max},p_{min}}$in two steps.   Let $q$ be the distribution on support $[c]$ with $c=1+ \lceil \log_{1+\eps} p_{max} - \log_{1+\eps}p_{min} \rceil $  that is defined by $q(i)=(1+\eps)^{i-1} \frac{\eps}{(1+\eps)^{c}-1}.$  The distribution  $f_{p,\eps,p_{max},p_{min}}$  has support $[cn]$, and for $i\in[n]$ and $j \in [c]$ it assigns probability $p(i)q(j)$ to domain element $c(i-1)+j.$
\end{definition}

It is convenient for us to view the $\mod r$ operator as giving an output in $[r]$, so that
``$r \mod r$'' equals $r.$

\begin{definition} \label{def:lb2} We define the distribution $g_{k, p,\eps,p_{max},p_{min}}$ from distribution  $f_{p,\eps,p_{max},p_{min}}$ of support $[m]$ via the following process.  Let $r=\lceil \frac{m}{k} \rceil,$ and let $a_1:=1$, and for all $i \in \{2,\ldots,r \},$ let $a_i:=\lceil (1+\eps)a_{i-1} \rceil.$  For each $i \in [m],$  we assign probability $\frac{f_{p,\eps,p_{max},p_{min}}(i)}{a_{i \mod r}}$ to each of the $a_j$ support elements in the set $\{1+t,2+t,\ldots, a_{i \mod r}+t\},$ where $t=\littlesum_{\ell=1}^{i-1} a_{(\ell \mod r)}.$
\end{definition}

\begin{lemma}~\label{lb:samp}
Given $\eps,p_{min},p_{max},$ and access to independent samples from distribution $p$, one can generate independent samples from $f_{p,\eps,p_{max},p_{min}}$ and from  $g_{k, p,\eps,p_{max},p_{min}}.$
\end{lemma}
\begin{proof}
To generate a sample according to $f_{p,\eps,p_{max},p_{min}},$ one simply takes a sample $i \leftarrow p$ and then draws $j \in [c]$ according to the distribution $q$ as defined in Definition~\ref{def:lb} (note that this draw according to $q$ only involves  $\eps, p_{min}$ and $p_{max}$).  We then output the value $c(i-1)+j.$  It follows immediately from the above definition that the distribution of the output value is $f_{p,\eps,p_{max},p_{min}}.$

To generate a sample according to $g_{k,p,\eps,p_{max},p_{min}}$ given a sample $i \leftarrow f_{p,\eps,p_{max},p_{min}},$ one simply outputs (a uniformly random) one of the $a_{(i \mod r)}$ support  elements of $g_{k,p,\eps,p_{max},p_{min}}$ corresponding to the element $i$ of $f_{p,\eps,p_{max},p_{min}}.$  Specifically, if the support of $f_{p,\eps,p_{max},p_{min}}$ is $[m]$, then we output a random element of the set $\{1+t,2+t,\ldots, a_{i \mod r}+t\},$ where $t=\littlesum_{\ell=1}^{i-1} a_{(\ell \mod r)},$ with $a_j$ as defined in Definition~\ref{def:lb2}, and $r=\lceil \frac{m}{k} \rceil.$
\end{proof}

\begin{lemma}~\label{lb:kmod}
If $p_{min} \leq p(i) \leq p_{max}$ for all $i \in [n]$,  then the distribution $f_{p,\eps,p_{max},p_{min}}$ of Definition~\ref{def:lb}, with density $f:[cn]\rightarrow \RR,$ has the property that $\frac{f(i)}{f(i-1)}\le 1+\eps$ for all $i>1$, and the distribution  $g_{k,p,\eps,p_{max},p_{min}}$ of Definition~\ref{def:lb2} is $2(k-1)$-modal.
\end{lemma}
\begin{proof}
Note that the distribution $q$, with support $[c]$ as defined in Definition~\ref{def:lb},  has the property that $q(i)/q(i-1)=1+\eps$ for all $i \in \{2,\ldots,c\},$ and thus $f({\ell})/f({\ell}-1)=1+\eps$ for any ${\ell}$ satisfying $({\ell} \mod c) \neq 1.$  For values $\ell$ that are 1 mod $c$, we have $$\frac{f({\ell})}{f({\ell}-1)} =\frac{p(i+1)}{p(i)(1+\eps)^{c-1}} \le \frac{p(i+1)p_{min}}{p(i)p_{max}} \le 1.$$

Given this property of $f_{p,\eps,p_{max},p_{min}},$ we now establish that  $g_{k, p,\eps,p_{max},p_{min}}$ is monotone decreasing on each of the $k$ equally sized contiguous regions of its domain.  First consider the case $k=1$; given a support element $j$, let $i$ be such that $j \in \{1+\littlesum_{\ell=1}^{i-1}a_\ell,\ldots, a_i+\littlesum_{\ell=1}^{i-1}a_\ell\}.$ We thus have that $$g_{1,p,\eps,p_{max},p_{min}}(j) = \frac{f_{p,\eps,p_{max},p_{min}}(i)}{a_i} \le \frac{(1+\eps)f_{p,\eps,p_{max},p_{min}}(i-1)}{a_i} \le \frac{f_{p,\eps,p_{max},p_{min}}(i-1)}{a_{i-1}} \le g_{1,p,\eps,p_{max},p_{min}}(j-1),$$
and thus $g_{1,p,\eps,p_{max},p_{min}}$ is indeed 0-modal since it is monotone non-increasing.
For $k>1$ the above arguments apply to each of the $k$ equally-sized contiguous regions of the support, so there are $2(k-1)$ modes, namely the local maxima occurring at the right endpoint of each region, and the local minima occurring at the left endpoint of each region.
\end{proof}

\begin{lemma}~\label{lb:l1pres}
For any distributions $p,p'$ with support $[n]$, and any $\eps, p_{max},p_{min},$ we have that
$$\dtv(p,p')= \dtv\left( f_{p,\eps,p_{max},p_{min}},f_{p',\eps,p_{max},p_{min}}\right)=\dtv\left( g_{k,p,\eps,p_{max},p_{min}},g_{k,p',\eps,p_{max},p_{min}}\right).$$
\end{lemma}
\begin{proof}
Both equalities follow immediately from the fact that the transformations of Definitions~\ref{def:lb} and ~\ref{def:lb2} partition each element of the input distribution in a manner that is oblivious to the probabilities.  To illustrate, letting $c=1+\lceil \log_{1+\eps} p_{max} - \log_{1+\eps}p_{min} \rceil,$ and letting $q$ be as in Definition~\ref{def:lb}, we have the following:
\begin{eqnarray*}
\dtv\left( f_{p,\eps,p_{max},p_{min}},f_{p',\eps,p_{max},p_{min}}\right) & = & \littlesum_{i \in [n], j \in [c]} q(j)|p(i)-p'(i)| \\
& = &\littlesum_{i \in [n]} |p(i)-p'(i)|.
\end{eqnarray*}
\end{proof}

\begin{lemma}~\label{lb:size}
If $p$ has support $[n],$ then for any $\eps<1/2,$ the distribution $g_{k,p,\eps,p_{max},p_{min}}$ is
supported on $[N]$, where $N$ is at most $k \frac{ e^{\frac{8n}{k}\left(1+ \log(p_{max}/p_{min}) \right)}}{\eps^2}$.
\end{lemma}
\begin{proof}
The support of  $f_{p,\eps,p_{max},p_{min}}$ is $n(1+\lceil \log_{1+\eps}p_{max}-\log_{1+\eps}p_{min} \rceil ) \le n\left( 2+ \frac{\log(p_{max}/p_{min})}{\log(1+\eps)}\right).$  Letting $a_1:=1$ and $b_1:= \lceil \frac{1}{\eps} \rceil,$ and defining $a_i:= \lceil a_{i-1}(1+\eps)\rceil,$ and $b_i:= \lceil b_{i-1}(1+\eps)\rceil,$ we have that $a_i \le b_i$ for all $i$.  Additionally, $b_{i+1}/b_i \le 1+2\eps$, since all $b_i \ge 1/\eps,$ and thus the ceiling operation can increase the value of $(1+\eps)b_i$ by at most $\eps b_i$.  Putting these two observations together, we have $$\littlesum_{i=1}^m a_i \le \littlesum_{i=1}^m b_i \le \frac{(1+2\eps)^{m+1}}{2 \eps^2}.$$ For any $\eps \le 1/2,$ we have that the support of $g_{k,p,1/2,p_{max},p_{min}}$ is at most
\begin{eqnarray*}
k\frac{(1+2\eps)^{\left \lceil \frac{n}{k}\left(2+ \frac{\log(p_{max}/p_{min})}{\log(1+\eps)}\right) \right\rceil}}{\eps^2} & \le & k\frac{(1+2\eps)^{2 \frac{n}{k}\left(2+ 4 \frac{\log(p_{max}/p_{min})}{2 \eps}\right) }}{\eps^2} \\
& \le & k \frac{ (1+2\eps)^{\frac{1}{2\eps} \left( \frac{8n}{k}\left(1+ \log(p_{max}/p_{min})\right) \right)}}{\eps^2} \\
& \le & k \frac{ e^{\frac{8n}{k}\left(1+ \log(p_{max}/p_{min}) \right)}}{\eps^2}.
\end{eqnarray*}
\end{proof}

\begin{proof} [Proof of Proposition~\ref{prop:lb}]
The proof is now a simple matter of assembling the above parts. Given a distribution $D$ over pairs of distributions of support $[n],$ as specified in the proposition statement, the distribution $D_k$ is defined via the process of taking $(p,p')\leftarrow D,$ then applying the transformation of Definitions~\ref{def:lb} and ~\ref{def:lb2} with $\eps=1/2$ and to yield a pair $\left(g_{k,p,1/2,p_{max},p_{min}},g_{k,p',1/2,p_{max},p_{min}}\right)$.  We claim that this $D_k$ satisfies all the properties claimed in the proposition statement.   Specifically, Lemmas~\ref{lb:kmod} and ~\ref{lb:size}, respectively, ensure that every distribution in the support of $D_k$ has at most $2(k-1)$ modes, and has support size at most $4 ke^{\frac{8n}{k}\left(1+ \log(p_{max}/p_{min}) \right)}.$  Additionally, Lemma~\ref{lb:l1pres} guarantees that the transformation preserves $L_1$ distance, namely, for two distributions $p,p'$ with support $[n],$ we have $L_1(p,p')=L_1(g_{k,p,1/2,p_{max},p_{min}},g_{k,p',1/2,p_{max},p_{min}}).$  Finally, Lemma~\ref{lb:samp} guarantees that, given $s$ independent samples from $p$, one can simulate drawing $s$ independent samples according to $g_{k,p,1/2,p_{max},p_{min}}.$  Assuming for the sake of contradiction that one had an algorithm that could distinguish whether $L_1(g_{k,p,1/2,p_{max},p_{min}},g_{k,p',1/2,p_{max},p_{min}})$ is  less than $\eps_1$ versus greater than $\eps_2$ with the desired probability given $s$ samples, one could take $s$ samples from distributions $(p,p')\leftarrow D$, \emph{simulate} having drawn them from $g_{k,p,1/2,p_{max},p_{min}}$ and $g_{k,p',1/2,p_{max},p_{min}},$ and then run the hypothesized tester algorithm on those samples, and output the answer, which will be the same for $(p,p')$ as for $(g_{k,p,1/2,p_{max},p_{min}},g_{k,p',1/2,p_{max},p_{min}}).$  This contradicts the assumption that no algorithm with these success parameters exists for $(p,p')\leftarrow D.$
\end{proof}

\section{Proof of Theorem~\ref{thm:birge-oblivious}} \label{ap:birge}

We first note that we can assume that $\eps > 1/ n$. Otherwise, the decomposition of $[n]$ into singleton intervals
$I_i  = \{ i\}$, $i \in [n]$, trivially satisfies the statement of the theorem. Indeed, in this case we have that $(1/\eps) \cdot \log n >n$ and $p_f \equiv p$.


We first describe the oblivious decomposition and then show that it satisfies the statement of the theorem.
The decomposition $\mathcal{I}$ will be a partition of $[n]$ into $\ell$ nonempty consecutive intervals $I_1, \ldots, I_{\ell}$.
In particular, for $j\in[\ell]$, we have $I_{j} = [n_{j-1}+1, n_j]$ with $n_0=0$ and $n_{\ell} = n$.
The {\em length} of interval $I_i$,  denoted by $l_i$, is defined to be the cardinality of $I_i$, i.e., $l_i = |I_i|$.
(Given that the intervals are disjoint and consecutive, to fully define them it suffices to specify their lengths.)

We can assume wlog that $n$ and $1/\eps$ are each at least sufficiently large universal constants.
The interval lengths are defined as follows. Let $\ell \in \mathbb{Z}^+$ be the smallest integer such that
$$\littlesum_{i=1}^{\ell} \lfloor (1+\eps)^i \rfloor \geq n.$$
For $i=1,2,\ldots, \ell-1$ we define
$$ l_i  := \lfloor (1+\eps)^i \rfloor .$$
For the $\ell$-th interval, we set
$$l_{\ell} :=  n - \littlesum_{i=1}^{\ell-1} l_i.$$

It follows from the aforementioned definition that the number $\ell$ of intervals in the decomposition
is at most \ignore{$O((1/\eps) \cdot \log n)$. The calculation actually gives that}
$$O\left( (1/\eps) \cdot \log (1+\eps \cdot n) \right).$$


Let $p$ be any non-increasing distribution over $[n]$. We will now show that the above described decomposition satisfies
$$ \dtv(p_f, p)  = O(\eps)$$
where $p_f$ is the flattened distribution corresponding to $p$  and the partition $\mathcal{I} = \{I_i\}_{i=1}^{\ell}$.
We can write
$$\dtv(p_f, p) = (1/2) \cdot \littlesum_{i=1}^n \left| p_f (i) - p(i) \right|  = \littlesum_{j=1}^{\ell}  \dtv \left( (p_f)^{I_j}, p^{I_j} \right) $$
where $p^I$ denotes the (sub-distribution) restriction of $p$ over $I$.

Let $I_j = [n_{j-1}+1, n_j]$ with $l_j = |I_j| = n_j - n_{j-1}$. Then we have that
$$\dtv \left( (p_f)^{I_j}, p^{I_j} \right) = (1/2) \cdot \littlesum_{i=n_{j-1}+1}^{n_j} \left| p_f(i) - p(i) \right| .$$
Recall that $p_f$ is by definition constant within each $I_j$ and in particular equal to $\bar{p}_f^j =  \littlesum_{i=n_{j-1}+1}^{n_j} p(i) / l_j$.
Also recall that $p$ is non-increasing, hence $p(n_{j-1}) \geq p(n_{j-1}+1) \geq \bar{p}_f^j \geq p(n_{j})$. Therefore, we can bound from above the variation distance
within $I_j$ as follows
$$\dtv \left( (p_f)^{I_j}, p^{I_j} \right)  \leq l_j \cdot \left( p(n_{j-1}+1) - p(n_j) \right) \leq  l_j \cdot \left( p(n_{j-1}) - p(n_j) \right).$$
So, we have
\begin{equation}
\dtv(p_f, p) \leq   \littlesum_{j=1}^{\ell} l_j \cdot \left( p(n_{j-1}) - p(n_j) \right).
\label{eq:pfvsp}
\end{equation}

To bound the above quantity we analyze summands with $l_j < 1/\eps$ and with
$l_j \geq 1/\eps$ separately.\ignore{proceed as follows.
We break the contribution into two parts: the one that comes
from intervals of length at most $1/\eps$, and the one that comes from the rest.}

Formally, we partition the set of intervals $I_1, \ldots, I_{\ell}$ into ``short'' intervals and ``long intervals'' as follows:
If any interval $I_j$ satisfies $l_j \geq 1/\eps$, then let $j_0 \in \mathbb{Z}^+$ be the largest integer  such that $l_{j_0} < 1/\eps$; otherwise we have that every
interval $I_j$ satisfies $l_j < 1/\eps$, and in this case we let $j_0=\ell.$
If $j_0 < \ell$ then we have that $j_0 = \Theta ((1/\eps) \cdot \log_2 (1/\eps))$.
Let $S=\{I_i\}_{i=1}^{j_0}$ denote the set of \emph{short} intervals and let $L$ denote its complement $L = \mathcal{I} \setminus S$.





Consider the short intervals and cluster them into {\em groups} according to their length; that is, a group contains all intervals in $S$ of the same length.
We denote by $G_i$ the $i$th group, which by definition contains all intervals in $S$
of length $i$; note that these intervals are consecutive.
The {\em cardinality} of a group (denoted by $|\cdot|$) is the number of intervals it contains;
the {\em length} of a group is the number of elements it contains (i.e. the sum of the lengths of
the intervals it contains).

\ignore{
\inote{Again, we note that some groups may be empty, e.g. if $\eps = 100/ n$, we will only have groups of at most constant size,
if $\eps = \log n / n$, we will have groups of size up to $\log(1/\eps)$, etc.}
}

Note that $G_1$ (the group containing all singleton intervals) has $|G_1| = \Omega(1/\eps)$ (this follows from the assumption that $1/\eps <n$).  Hence $G_1$ has length $\Omega (1/\eps)$. Let $j^{\ast} < 1/\eps$ be the maximum length of any short interval in $S$.  It is easy to verify that each group $G_j$ for $j \leq j^{\ast}$ is nonempty, and that for all $j \leq j^{\ast}-1$, we have $ |G_j| = \Omega \left( (1/\eps) \cdot (1/j) \right)$, \inote{Please check}\rnote{I convinced myself this is indeed correct.} which implies that the length
of $G_j$ is $\Omega(1/\eps)$.

To bound the contribution to (\ref{eq:pfvsp}) from the short intervals, we consider the corresponding sum for each group, and use the fact that $G_1$ makes no contribution
to the error. In particular, the contribution of the short intervals is

\begin{equation} \littlesum_{l = 2}^{j^{\ast}} l \cdot \left( p^{-}_l - p^{+}_l \right)
\label{eq:shortcontrib}
\end{equation}
where $p^{-}_l$ (resp.  $p^{+}_l$) is the probability mass of the leftmost (resp. rightmost) point in $G_l$.
Given that $p$ is non-increasing, we have that $p^{+}_l \geq p^{-}_{l+1}$.
Therefore, we can upper bound (\ref{eq:shortcontrib}) by
$$ 2 \cdot p^{+}_1 + \littlesum_{l=2}^{j^{\ast}-1} p^{+}_l - j^{\ast} \cdot p^{+}_{j^{\ast}}.$$
Now note that $p^{+}_1 = O(\eps) \cdot p(G_1)$, since $G_1$ has length (total number of elements) $\Omega(1/\eps)$
and $p$ is non-increasing. Similarly, for $l < j^{\ast}$, we have that
$p^{+}_l  = O(\eps) \cdot p(G_l)$, since  $G_l$ has length $\Omega(1/\eps)$.
Therefore, the above quantity can be upper bounded by
\begin{equation}O(\eps) \cdot p(G_1)  + O(\eps) \cdot \littlesum_{l=2}^{j^{\ast}-1} p(G_l)  - j^{\ast} \cdot p^{+}_{j^{\ast}} = O(\eps) \cdot p(S) -  j^{\ast} \cdot p^{+}_{j^{\ast}}.
\label{eq:ub}
\end{equation}
We consider two cases: The first case is that $L=\emptyset$. In this case, we are done because the above expression (\ref{eq:ub}) is $O(\eps)$.
The second case is that $L \neq \emptyset$ (we note in passing that in this case
the total number of elements in all short intervals is $\Omega(1/\eps^2)$, which means
that we must have $\eps = \Omega(1/\sqrt{n})$).  In this case
we bound the contribution of the long intervals using the same argument as Birg\'{e}.
In particular, the contribution of the long intervals is
  \begin{equation}
  \littlesum_{j= j_0 + 1}^{\ell} l_j \cdot \left( p(n_{j-1}) - p(n_j) \right) \leq (j^{\ast}+1) \cdot p^{+}_{j^{\ast}} +   \littlesum_{j= j_0 + 1}^{\ell-1} (l_{j+1} - l_{j}) \cdot p(n_{j}).
  \label{eq:smoob} \end{equation}
\rnote{I don't follow why the above inequality holds; can we explain more?}
\inote{My vague recollection: Telescoping sum; first term does not cancel out. Since this is the first long interval,
its length is at most by one bigger than the previous length. By monotonicity, $p(n_{j_0})$ is at most  $ p^{+}_{j^{\ast}}.$}

Given that $l_{j+1} - l_j \leq (2\eps) \cdot l_j$ and $\littlesum_j l_j \cdot p(n_{j}) \leq p(L)$, it follows that the second summand in (\ref{eq:smoob}) is at most $O(\eps) \cdot p(L)$.
Therefore, the total variation distance between $p$ and $p_f$ is at most (\ref{eq:ub}) +
(\ref{eq:smoob}), i.e.
\begin{equation}
O(\eps) \cdot p(S) + O(\eps) \cdot p(L) +  p^{+}_{j^{\ast}} .
\label{eq:stoot}
\end{equation}
 Finally, note that $p(L)+p(S)=1$ and $p^{+}_{j^{\ast}}  = O(\eps)$. (The latter holds because $p^{+}_{j^{\ast}}$ is the probability mass of the rightmost point in $S$;
recall that $S$ has length at least $1/\eps$ and $p$ is decreasing.)  This implies that (\ref{eq:stoot}) is at most $O(\eps)$, and this completes the proof
of Theorem~\ref{thm:birge-oblivious}.

\section{Bounding variation distance} \label{ap:boundvar}

As noted above, our tester will work by decomposing the interval
$[n]$ into sub-intervals.  The following lemma will be useful
for us; it bounds
the variation distance between two distributions $p$ and $q$ in terms
of how $p$ and $q$ behave on the sub-intervals in such a decomposition.

\begin{lemma} \label{lem:decomp}
Let $[n]$ be partitioned into $I_1,\dots,I_r$.  Let $p,q$ be two distributions over $[n]$.  Then
\begin{equation}
\label{eq:bound}
\dtv(p,q) \leq {\frac 1 2} \littlesum_{j=1}^r |p(I_j) - q(I_j)| + \littlesum_{j=1}^r p(I_j) \cdot \dtv(p_{I_j},q_{I_j}).
\end{equation}
\end{lemma}

\begin{proof}
Recall that $\dtv(p,q) = {\frac 1 2} \littlesum_{i=1}^n |p(i)-q(i)|.$
To prove the claim it suffices to show that
\begin{equation}
{\frac 1 2} \littlesum_{i \in I_1}|p(i)-q(i)| \leq {\frac 1 2}
|p(I_1) - q(I_1)| + p(I_1) \cdot \dtv(p_{I_1},q_{I_1}).
\label{eq:a}
\end{equation}

We assume that $p(I_1) \leq q(I_1)$ and prove (\ref{eq:a})
under this assumption.  This gives
the bound in general since if $p(I_1) > q(I_1)$ we have
\[
{\frac 1 2} \littlesum_{i \in I_1}|p(i)-q(i)| \leq |p(I_1) - q(I_1)| + q(I_1) \cdot \dtv(p_{I_1},q_{I_1})
< |p(I_1) - q(I_1)| + p(I_1) \cdot \dtv(p_{I_1},q_{I_1})
\]
where the first inequality is by (\ref{eq:a}).
The triangle inequality gives us
\[
|p(i)-q(i)| \leq \left|p(i) - q(i) \cdot {\frac {p(I_1)} {q(I_1)}}\right| +
\left|q(i) \cdot {\frac {p(I_1)} {q(I_1)}} - q(i)\right|.
\]
Summing this over all $i \in I_1$ we get
\[
{\frac 1 2} \littlesum_{i \in I_1}|p(i)-q(i)| \leq {\frac 1 2} \littlesum_{i \in I_1}
\left|p(i) - q(i) \cdot {\frac {p(I_1)} {q(I_1)}}\right| + {\frac 1 2} \littlesum_{i \in I_1}
\left|q(i) \cdot {\frac {p(I_1)} {q(I_1)}} - q(i)\right|.
\]
We can rewrite the first term on the RHS as
\begin{eqnarray*}
{\frac 1 2} \littlesum_{i \in I_1}
\left|p(i) - q(i) \cdot {\frac {p(I_1)} {q(I_1)}}\right|&=&
p(I_1) \cdot {\frac 1 2} \littlesum_{i \in I_1}
\left|{\frac {p(i)}{p(I_1)}} - {\frac {q(i)}{q(I_1)}}\right|=
p(I_1) \cdot {\frac 1 2} \littlesum_{i \in I_1}
\left|p_{I_1}(i) - q_{I_1}(i)\right|\\
&=& p(I_1) \cdot \dtv(p_{I_1},q_{I_1})
\end{eqnarray*}
so to prove the desired bound it suffices to show that
\begin{eqnarray}
{\frac 1 2} \littlesum_{i \in I_1}
\left|q(i) \cdot {\frac {p(I_1)} {q(I_1)}} - q(i)\right| \leq
|p(I_1) - q(I_1)|.
\label{eq:b}
\end{eqnarray}
We have
\[
\left|q(i) \cdot {\frac {p(I_1)} {q(I_1)}} - q(i)\right| = q(i) \cdot
\left|{\frac {p(I_1)} {q(I_1)}} - 1\right|
\]
and hence we have
\[
{\frac 1 2} \littlesum_{i \in I_1} \left|
q(i) \cdot {\frac {p(I_1)} {q(I_1)}} - q(i)\right| = {\frac 1 2} \littlesum_{i \in I_1}
q(i) \cdot
\left|{\frac {p(I_1)} {q(I_1)}} - 1\right| = {\frac 1 2}
q(I_1) \cdot \left|{\frac {p(I_1)} {q(I_1)}} - 1\right| = {\frac 1 2}
|p(I_1) - q(I_1)|.
\]
So we indeed have (\ref{eq:b}) as required, and the lemma holds.
\end{proof}

%

\end{document}